\documentclass[preprint,amsmath,amssymb,floatfix,nofootinbib]{revtex4-2}

\usepackage[T1]{fontenc}
\usepackage[utf8]{inputenc}
\usepackage{lmodern}
\usepackage{microtype}            
\usepackage{amsthm}
\theoremstyle{definition}

\newtheorem{lemma}{Lemma}
\newtheorem{proposition}{Proposition}
\theoremstyle{remark}
\newtheorem*{remark}{Remark}
\usepackage{booktabs}
\usepackage{graphicx}
\usepackage{placeins}
\usepackage{xcolor}
\definecolor{linknavy}{RGB}{14,54,110}
\usepackage[colorlinks=true,linkcolor=linknavy,citecolor=linknavy,urlcolor=linknavy]{hyperref}


\newcommand{\R}{\mathbb{R}}
\newcommand{\cM}{\mathcal{M}}
\newcommand{\hcM}{\widehat{\mathcal{M}}}
\newcommand{\cN}{\mathcal{N}}
\newcommand{\hatH}{\widehat{H}}
\newcommand{\hatPhi}{\widehat{\Phi}}
\newcommand{\Ctau}{\mathcal{C}}
\newcommand{\Pihstep}{\Phi^{\tau}_{\mathrm{Pih}}}
\newcommand{\numstep}{\Phi^{\tau}_{\mathrm{num}}}
\newcommand{\Hflow}{\Phi^{\tau}_{H}}
\newcommand{\heta}{\hat{\eta}}
\newcommand{\Hc}{H_{\mathrm{c}}}
\newcommand{\Emech}{E_{\mathrm{mech}}}
\newcommand{\rhoeta}{\rho_{\eta}}

\begin{document}

\title{A Projected Semiexplicit Integrator for Dissipative Systems with
  Configuration-Dependent Kinetic Energy:\\
  Contact-Herglotz Formulation and Benchmarks}

\author{Lorena Loera-Galeana}
\email{loera\_l@tec.mx}
\author{Santiago Mejía}
\email{A01751866@tec.mx}
\author{Espartaco Alvarado}
\email{A00839913@tec.mx}
\author{Héctor Medel-Cobaxin}
\email[Corresponding author: ]{hmedel@tec.mx}

\affiliation{Tecnologico de Monterrey, Escuela de Ingenieria y Ciencias,
  Ave. Eugenio Garza Sada 2501 Sur, Col. Tecnologico,
  Monterrey, N.L. 64700, Mexico}

\date{August 17, 2026}

\begin{abstract}
Contact Hamiltonian dynamics gives dissipative mechanics an intrinsic action
variable, but explicit contact splittings reach only those kinetic energies
whose terms are exactly integrable: this includes diagonal
configuration-dependent metrics with frozen-coordinate coefficients (the
spherical pendulum, a particle on a torus), and it excludes dense metrics with
momentum cross terms, of which the double pendulum is the flagship. We introduce
a projected Pihajoki-contact integrator for this non-separable setting,
combining phase-space duplication, symmetric projection onto the physical
diagonal, and constant-friction damping half-steps, with the action factor
carried by an exact Herglotz update. As in the projected
extended-phase-space framework it builds on, the construction needs no binding
parameter, returns the duplicated copies to the physical diagonal exactly at
every step, and confines any nonlinear solve to the $2n$ projection variables.
For constant friction the step rescales $\omega=d\eta$ by the exact factor
$e^{-\gamma\tau}$ when the projection is solved exactly (a classical
conformally symplectic identity, realized here for this class), while
time-symmetry, consistency, and smoothness yield an $\mathcal{O}(\tau^3)$
one-step contact-form residual, a bound not specific to the contact form. On
the damped double pendulum, spherical pendulum, and torus particle the method
is second-order accurate, reproduces the contact decay law, and controls long-time
energy and contact drift in coarse or stiff regimes where both the Tao
baseline and the unprojected symmetric average lose the solution. A
head-to-head with exact-contactomorphism splittings delimits the niche: where
a frozen-coordinate splitting exists it preserves the contact form exactly and
wins at matched cost; for the dense double-pendulum metric the explicit
splitting realizable today is first-order with a prohibitive error constant,
and the exactly conformal second-order alternative we construct (implicit
midpoint on the kinetic map, lifted through its generating function) is fully
implicit, so what the projected method offers on dense metrics is second-order
accuracy with the solve confined to the projection. The contact-form estimate is local, one-step, and
constant-friction.
\end{abstract}

\maketitle


\section{Introduction}
\label{sec:introduction}

Dissipation is ubiquitous in mechanics but does not fit naturally within the
symplectic formulation: friction, drag, and internal damping produce irreversible
energy loss, classically modeled by external constructions such as Rayleigh
dissipation, port-Hamiltonian couplings, or GENERIC formalisms~\cite{Carinena2024},
which need not preserve the variational character of the equations. Contact
geometry provides an intrinsic alternative: extending the phase space with an
action variable $z$ represents energy loss geometrically, closely tied to the
Herglotz variational principle in which the action solves a differential
equation~\cite{Herglotz1930,deLeon2021review}. The contact Hamiltonian and
Lagrangian formalisms and their symmetry properties are by now well
developed~\cite{Bravetti2017,deLeon2019}, and they provide the geometric
language in which we frame both the numerical method and the diagnostics used to
assess it.

This structure has motivated geometry-preserving integrators. Contact variational
and splitting schemes~\cite{McLachlanQuispel2002} are effective when the contact
Hamiltonian is \emph{separable}, $H=T(p)+V(q)+\gamma(q)z$: the sub-flows integrate
explicitly and symmetric compositions give second-order contact methods that reduce
to symplectic schemes as
$\gamma\to0$~\cite{Vermeeren2019,Bravetti2020,Zadra2021,Bravetti2021}, covering the
damped oscillator, the simple pendulum, and other one-degree-of-freedom models. A
complementary, variational route builds contact integrators from a discrete Herglotz
principle rather than from an operator splitting: the discrete contact mechanics
of~\cite{Simoes2021} yields integrators that are conformally contact by construction,
specializing to the contact setting the discrete-variational treatment of forced and
dissipative systems~\cite{Kane2000}. Both the splitting and the variational
constructions, however, are formulated for separable or variationally defined
Lagrangians. Contact---more generally Jacobi---dynamics can also be lifted to
a homogeneous symplectic realization and discretized there by generating
functions~\cite{Araujo2026}, at the price of an implicit solve per step; the
reported examples have kinetic energies without configuration-dependent
coupling. A recent alternative route is that of
Kevrekidis~\cite{Kevrekidis2026}, who shows that the Lie algebra generated by
strict and prolonged contact Hamiltonians contains every Hamiltonian that is
polynomial in $p$ with coefficients depending on $(q,z)$, and builds high-order
contact splittings from that representation. This construction
covers the present class: $H=\tfrac12 g^{ij}(q)p_ip_j+V(q)+\gamma(q)z$ is
quadratic in $p$ with $(q,z)$-dependent coefficients, so it lies in that algebra
exactly, with no polynomial approximation, and the resulting sub-steps are exact
contactomorphisms. On our target class that construction is therefore
theoretically stronger than what we prove: it preserves the contact structure
exactly, where our estimate is $\mathcal{O}(\tau^3)$, local and one-step. What it
costs is a commutator representation whose composition requires many sub-steps and
both signs of the time step, and its reported demonstrations are
one-degree-of-freedom. The present method takes the opposite trade: a direct,
second-order, two-sub-step treatment of the geodesic kinetic term on the benchmark
geometries below, bought with a weaker theorem. The
difficulty is the \emph{non-separable} case, where $T(q,p)=\tfrac12 g^{ij}(q)p_ip_j$
couples positions and momenta through a configuration-dependent metric, as for the
double pendulum, spherical pendulum, and particles on Riemannian manifolds. The
explicit kinetic shear is then unavailable and contact splitting becomes implicit
or must be replaced.

Once explicit contact splitting is unavailable, a non-separable system must be
advanced by some other integrator, and the standard candidates each involve a
compromise. Implicit midpoint gives a symmetric second-order
reference but needs a full nonlinear solve each step and is not designed to
preserve the contact form in the discretization used here. Tao's extended
phase-space method~\cite{Tao2016} restores explicitness by duplicating the phase
space and binding the copies with an integrable rotation, but it reads off one copy
without returning the duplicated variables exactly to the constraint manifold.
The conservative side of this extended-phase-space program remains active:
explicit symplectic integrators for general non-separable Hamiltonians have
recently been shown to exist via an invariant submanifold of the extended
space~\cite{Mei2025}; no dissipative counterpart exists so far.

The same obstruction is visible in the conformally symplectic literature, which
predates the contact framing. Modin and
S\"oderlind~\cite{ModinSoderlind2011} integrate Hamiltonian systems on a
Riemannian configuration manifold perturbed by Rayleigh damping using exactly
the three-term composition (damping half-step, symplectic Strang step, damping
half-step) that reappears in Section~\ref{subsec:dissipation} below, and the
conformal rescaling identity such a composition satisfies is Theorem~3.3
of~\cite{Franca2020}. Neither the composition nor the identity is claimed here as
new. For the present purpose the relevant point is the hypothesis their method
carries: it is ``crucial'', they state, that one knows coordinates in which the
kinetic flow is explicitly computable, ``typically accomplished by choosing
Cartesian coordinates in which the inertia operator $M$ is independent of $q$'',
and the $q$-dependent case is relegated to an implicit Runge--Kutta
sub-step. The configuration-dependent metric is thus either transformed away
by a change of coordinates or handed to an implicit sub-step rather than
treated in its own chart; a semiexplicit treatment of that case is the gap
this paper addresses.

Our contribution is narrower than the ingredient list
suggests: we formulate representative non-separable dissipative
systems in contact-Herglotz form, and we construct and benchmark a semiexplicit
projected integrator for constant-friction non-separable contact Hamiltonians,
combining phase-space duplication~\cite{Pihajoki2015}, symmetric projection onto
the physical diagonal~\cite{Jayawardana2023}, constant-friction contact damping,
and a Herglotz
action update for $z$. None of these ingredients is new individually, and neither
is the conformal composition they are assembled
into~\cite{ModinSoderlind2011,Franca2020}; what we claim is their combination
into a single contact-Herglotz scheme for
$H=\tfrac12 g^{ij}(q)p_ip_j+V(q)+\gamma(q)z$ that acts directly on the
configuration-dependent metric rather than assuming it away, returns the copies
to the diagonal by symmetric averaging at every step, confines the nonlinear
solve to the $2n$ projection variables, and admits a local contact-conformal
estimate under constant friction. Sections~\ref{sec:background}
and~\ref{sec:reeb-lie} recall the
background and Reeb diagnostic; Section~\ref{sec:problem} the
separable/non-separable dichotomy and benchmark systems; Section~\ref{sec:pihajoki}
the method and its analytical results; Section~\ref{sec:comparison} the comparison
methods; Section~\ref{sec:experiments} the numerical experiments; and
Section~\ref{sec:discussion} the discussion and limitations.

\section{Contact Hamiltonian background}
\label{sec:background}

Contact Hamiltonian mechanics lives on the extended phase space
$\cM = T^*Q\times\R$ with Darboux coordinates $(q^i,p_i,z)$. The contact
structure is the one-form
\begin{equation}
  \eta = dz - p_i\,dq^i,
  \label{eq:contact_form}
\end{equation}
which satisfies $\eta\wedge(d\eta)^n\neq 0$, so that $(\cM,\eta)$ is a
$(2n+1)$-dimensional contact manifold.

\subsection{Contact Hamilton equations}

For a Hamiltonian $H\colon\cM\to\R$, the associated contact Hamiltonian vector field $X_H$ is defined
by $\iota_{X_H}\eta=-H$ and
$\iota_{X_H}d\eta=dH-(\partial H/\partial z)\,\eta$, yielding
\begin{align}
  \dot{q}^i &= \frac{\partial H}{\partial p_i}, \label{eq:ham_q}\\
  \dot{p}_i &= -\frac{\partial H}{\partial q^i}
               - p_i\,\frac{\partial H}{\partial z}, \label{eq:ham_p}\\
  \dot{z}   &= p_i\,\frac{\partial H}{\partial p_i} - H.
               \label{eq:ham_z}
\end{align}
When $\partial H/\partial z=0$ these reduce to the standard Hamilton equations,
with $z$ tracking the classical action along the trajectory.

\subsection{Energy decay}

Along solutions of \eqref{eq:ham_q}-\eqref{eq:ham_z}, 
differentiating the Hamiltonian shows that the symplectic cross-terms cancel, and
hence
\begin{equation}
  \dot{H} = -H\,\frac{\partial H}{\partial z}.
  \label{eq:energy_decay}
\end{equation}
For the mechanically standard choice $H=H_0(q,p)+\gamma z$ with $\gamma>0$
constant, this integrates to $H(t)=H(0)e^{-\gamma t}$, while
\eqref{eq:ham_p} becomes $\dot{p}_i= -\partial_{q^i}H_0-\gamma p_i$,
reproducing linear (Rayleigh-type) damping.

\subsection{Herglotz variational origin}

Equations~\eqref{eq:ham_q}-\eqref{eq:ham_z} can be obtained from the Herglotz
variational principle~\cite{Herglotz1930}. Given a Lagrangian $L(q,\dot q,z)$, define the action variable $z(t)$
by $\dot z=L(q,\dot q,z)$, with initial condition $z(0)=z_0$. One then extremizes the final value $z(T)$ over curves with
fixed endpoints. Stationarity yields the Herglotz-Euler-Lagrange equation
\begin{equation}
  \frac{d}{dt}\frac{\partial L}{\partial\dot{q}^i}
  -\frac{\partial L}{\partial q^i}
  = \frac{\partial L}{\partial\dot{q}^i}\,\frac{\partial L}{\partial z},
  \label{eq:hel}
\end{equation}
which the Legendre transform $p_i=\partial L/\partial\dot{q}^i$,
$H=p_i\dot q^i-L$, maps exactly onto \eqref{eq:ham_q}-\eqref{eq:ham_z}.

The variable $z$ is therefore not an arbitrary auxiliary coordinate but the
accumulated action, determined by $\dot z=L(q,\dot q,z)$: for
$H=E_{\mathrm{mech}}+\gamma z$ the full contact Hamiltonian obeys the contact
decay law while $E_{\mathrm{mech}}$ and $\gamma z$ account for the exchange
between mechanical energy and dissipated action, so $z$ records the action
associated with dissipation rather than serving as a posterior diagnostic.

\section{Reeb fields and the numerical contact diagnostic}
\label{sec:reeb-lie}

This section records only the contact-geometric identities used later to define
the numerical diagnostics, without rederiving the standard Reeb-field calculus;
the conventions are those of Bravetti-Cruz-Tapias and de~Le\'on-Lainz
Valc\'azar~\cite{Bravetti2017,deLeon2019,deLeon2021review}. With the sign
convention of Section~\ref{sec:background}, the contact Hamiltonian vector field
is characterized by
\begin{equation}
  \iota_{X_H}\eta=-H,
  \qquad
  \iota_{X_H}d\eta=dH-R(H)\eta,
  \label{eq:contact_vector_field_reeb}
\end{equation}
where $R$ is the Reeb vector field defined by
\begin{equation}
  \eta(R)=1,\qquad \iota_R d\eta=0.
  \label{eq:reeb_def}
\end{equation}
The standard consequence is
\begin{equation}
  \mathcal{L}_{X_H}\eta=-R(H)\eta,
  \qquad
  X_H(H)=-H\,R(H).
  \label{eq:energy_decay_reeb}
\end{equation}
These identities provide the two diagnostics used below: the conformal factor
of the contact form and the scalar Hamiltonian decay law checked in the
numerical experiments.

\subsection{Specialization to the benchmarks and the duplicated space}

On the physical contact manifold with $\eta$ as in~\eqref{eq:contact_form} the
Reeb field is vertical,
\begin{equation}
  R=\frac{\partial}{\partial z},
  \label{eq:reeb_physical}
\end{equation}
so the
damped double pendulum, spherical pendulum, and Riemannian particle share the same
Reeb direction and differ only through their Hamiltonians and metrics. For
$H=E_{\mathrm{mech}}(q,p)+\gamma z$ with constant $\gamma$, $R(H)=\gamma$ and
\eqref{eq:energy_decay_reeb} specializes to
\begin{equation}
\begin{aligned}
  \mathcal{L}_{X_H}\eta &= -\gamma\eta, &
   (\Phi_H^t)^*\eta &= e^{-\gamma t}\eta,\\
  H(t) &= H(0)e^{-\gamma t}.
\end{aligned}
  \label{eq:lie_derivative_contact}
\end{equation}
The scalar relation $\dot H=-\gamma H$ is the decay diagnostic of
Section~\ref{sec:experiments}: along a damped run we compare
$\Hc(t)=\Emech(t)+\gamma z(t)$ with $\Hc(0)e^{-\gamma t}$ (the
conservative case $R(H)=0$ recovers strict preservation). The diagnostic is
gauge-fixed -- $V\to V+C$ shifts both $z(t)$ and $\Hc(0)$ -- so it is evaluated
with the potential convention used to generate the dynamics, and a decay check is
necessary but not sufficient for contact-form preservation.

The exact damping sub-flow used in the splitting is generated by
$H_\gamma=\gamma z$, giving $q\mapsto q$ and
$(p,z)\mapsto(e^{-\gamma t}p,\,e^{-\gamma t}z)$, with
\begin{equation}
  (\Phi_{H_\gamma}^t)^*\eta=e^{-\gamma t}\eta
  \label{eq:damping_lie_direct}
\end{equation}
(for $\gamma(q)$ the frozen-$q$ factor multiplying $\eta$ is $e^{-\gamma(q)t}$;
see the sub-flow~\eqref{eq:vargamma_step}, in which the momentum also picks up a
$z\,\partial_q\gamma$ deflection). The duplicated construction of
Section~\ref{sec:pihajoki} takes place on the extended contact form
$\heta=dz-p_a\,dq^a-y_a\,dx^a$, whose Reeb field is again $\partial/\partial z$.
Because the bare duplicated Hamiltonians are $z$-independent, their exact
contact flows are strict contact maps,
\begin{equation}
  (\hat\Phi_{A,\mathrm{c}}^t)^*\heta=\heta,\qquad
  (\hat\Phi_{B,\mathrm{c}}^t)^*\heta=\heta .
  \label{eq:extended_strict_contact}
\end{equation}
This identity must be read with care. The sub-steps actually composed in the
method---\eqref{eq:AB_flows} below---are not the contact flows
$\hat\Phi_{A,\mathrm{c}}^t$, $\hat\Phi_{B,\mathrm{c}}^t$
of~\eqref{eq:extended_strict_contact}: they freeze $z$, whose evolution is
instead reconstructed once per step by the Herglotz update~\eqref{eq:z_update}.
The maps~\eqref{eq:AB_flows} are the symplectic sub-flows of $\hatH_A$ and
$\hatH_B$ on $T^*Q\times T^*Q$, and they do not preserve $\heta$. Consequently
the extended contact structure plays no role in the conservative backbone and
enters none of the estimates of Section~\ref{subsec:accuracy}; those are
formulated on the physical contact manifold $(\cM,\eta)$ after projection.
Accordingly, on the diagonal embedding $\iota(q,p,z)=(q,p,q,p,z)$ one has
\begin{equation}
  \iota^*\heta=dz-2p_i\,dq^i
  \label{eq:diagonal_pullback_doubled}
\end{equation}
rather than $\eta$. The doubling in~\eqref{eq:diagonal_pullback_doubled} is a
normalization artifact and not a structural feature: replacing $\heta$ by
$dz-\tfrac12(p_a\,dq^a+y_a\,dx^a)$ leaves the Reeb field and the non-degeneracy
untouched and makes $\iota$ a strict contact embedding, $\iota^*\heta=\eta$. We
retain the unscaled convention for consistency with~\cite{Pihajoki2015} and
because only $d\heta$---the standard symplectic form on $\R^{4n}$---is used
below. Every contact-preservation statement in this paper is therefore made after
symmetric projection back to the physical variables and checked against $\eta$.

\section{The numerical problem}
\label{sec:problem}

The contact integrators we compare against in Section~\ref{sec:comparison}, and
the method of Section~\ref{sec:pihajoki}, all start from operator splitting
of the contact Hamiltonian. Whether this idea yields an explicit,
structure-preserving scheme depends on a single structural property: whether the
kinetic energy depends on the configuration. This section makes the dichotomy
precise. We first recall (Section~\ref{sec:separable}) why separable contact
Hamiltonians admit exact explicit sub-flows; we then explain
(Section~\ref{sec:nonsep}) why the same splitting breaks down for non-separable
Hamiltonians and why the standard repairs are unsatisfactory; finally
(Section~\ref{sec:benchmarks}) we introduce the three benchmark systems, the
double pendulum, the spherical pendulum, and the Riemannian particle, on which
the breakdown occurs and which drive the rest of the paper.

\subsection{Separable contact Hamiltonians: splitting is exact}
\label{sec:separable}

Consider a contact Hamiltonian of the separable form
\begin{equation}
\label{eq:separable}
  H(q,p,z) \;=\; T(p) + V(q) + \gamma(q)\, z ,
\end{equation}
in which the kinetic energy $T(p)$ depends on the momenta alone (the prototype
being $T(p)=|p|^{2}/2m$). Following the separable contact splitting
of~\cite{Bravetti2020,Zadra2021}, split
$H = H_{A}+H_{B}$
with
\begin{equation}
\label{eq:split}
  H_{A}(p) = T(p),
  \qquad
  H_{B}(q,z) = V(q) + \gamma(q)\, z .
\end{equation}

For the kinetic part, $\partial H_{A}/\partial z = 0$, so the contact Hamilton
equations for $H_{A}$ reduce to the symplectic kinetic flow with a decoupled
action update.
The momentum is conserved and the velocity $\dot q = \partial_{p}T$ is constant,
so the flow is an exact shear,
\begin{equation}
\label{eq:Aflow}
  \Phi^{h}_{A}:\;
  \begin{cases}
    q \mapsto q + h\,\partial_{p}T(p),\\
    p \mapsto p,\\
    z \mapsto z + h\,\bigl(p\cdot\partial_{p}T(p)-T(p)\bigr),
  \end{cases}
\end{equation}
explicit and exact for any $T(p)$, the action increment being the Legendre
transform $p\cdot\partial_pT-T$ dictated by~\eqref{eq:ham_z} rather than $T$
itself (for $T$ homogeneous of degree two in $p$, as for every benchmark of
this paper, Euler's identity collapses it to $h\,T(p)$; for non-quadratic $T$, such
as the relativistic $T=|p|$, the general form must be used).

Along the potential-dissipative flow of $H_{B}$ the configuration $q$ is
constant (because
$\partial H_{B}/\partial p = 0$), so $V(q)$ and $\gamma(q)$ are frozen and the
equations for $p$ and $z$ are linear ordinary differential equations solvable in
closed form. For constant $\gamma$,
\begin{equation}
\label{eq:Bflow}
\begin{aligned}
  z(h) &= e^{-\gamma h}\,z_{0} - \frac{V}{\gamma}\bigl(1 - e^{-\gamma h}\bigr),\\
  p_{i}(h) &= e^{-\gamma h}\,p_{i,0} - \frac{\partial_{i}V}{\gamma}\bigl(1 - e^{-\gamma h}\bigr).
\end{aligned}
\end{equation}
For position-dependent $\gamma(q)$ the $p$-equation acquires a non-autonomous
forcing $-(\partial_{i}\gamma)\,z(t)$, but, since $q$, and hence $\gamma$ and
$V$, remain frozen, variation of parameters still yields an exact closed-form
sub-flow (the explicit expression is recorded in
Section~\ref{sec:comparison}). Thus the isolated $H_B$ sub-flow remains
explicit; the local contact-conformal estimate for the projected non-separable method
below uses only the constant-friction case.

Combining the two exact sub-flows with the symmetric Strang composition
\begin{equation}
\label{eq:strang}
  \Phi^{\tau}_{\mathrm{split}} \;=\; \Phi^{\tau/2}_{B}\circ\Phi^{\tau}_{A}\circ\Phi^{\tau/2}_{B}
\end{equation}
gives a fully explicit, structure-preserving, second-order separable contact
splitting with no implicit solve (local truncation error $\mathcal{O}(\tau^{3})$
per step, global error $\mathcal{O}(\tau^{2})$; extensible to orders 4 and 6 by
Yoshida composition~\cite{Yoshida1990}). Because each sub-flow is the exact flow
of a contact Hamiltonian, and hence a contactomorphism, the composition is
itself a contactomorphism. In the standard mechanical case
$T(p)=\tfrac12p^\top M^{-1}p$ with constant $M$ and $\gamma\to 0$ it reduces to
the St\"ormer-Verlet (leapfrog) scheme for
$H_{0}=T+V$~\cite{Bravetti2020,Zadra2021}, while for a general separable $T(p)$
it remains a symmetric composition of exact sub-flows that should not be
identified with St\"ormer-Verlet; we reserve the term \emph{contact variational
integrator} for the discrete-Herglotz constructions
of~\cite{Vermeeren2019,Simoes2021}, a different object, derived by discretizing
the variational principle rather than by splitting the generator and
conformally contact by construction.

This regime is not narrow. It includes the damped harmonic oscillator, the
simple pendulum (whose kinetic energy $T=p_{\theta}^{2}/2m\ell^{2}$ is
independent of the angle), the Li\'enard equations~\cite{Zadra2021}, and the
Lane-Emden equation and the modified Kepler problem~\cite{Bravetti2020}.

In the present paper, the position-dependent case $\gamma(q)$
should be read as a modeling extension. As noted above, the contact-conformal
justification for the
projected method is restricted to constant $\gamma$, for which the damping
sub-flow has a scalar exponential factor and the contact decay law takes the
simple form $H(t)=H(0)e^{-\gamma t}$.

\subsection{Non-separable contact Hamiltonians: where splitting stops}
\label{sec:nonsep}

For a mechanical system whose configuration space is a Riemannian manifold
$(Q,g)$, or whose mass matrix depends on the configuration, the kinetic energy
takes the position-dependent form $T(q,p)=\tfrac12\,g^{ij}(q)\,p_{i}p_{j}$, and
the contact Hamiltonian
\begin{equation}
\label{eq:nonsep}
  H(q,p,z) \;=\; \tfrac12\, g^{ij}(q)\, p_i p_j + V(q) + \gamma(q)\, z
\end{equation}
is non-separable: positions and momenta are entangled in the kinetic term. The
potential-dissipative sub-flow $H_{B}$ is unaffected and remains exactly
solvable as in~\eqref{eq:Bflow}. The entire obstruction lies in the kinetic
sub-flow $H_{A}=T(q,p)$, whose contact Hamilton equations
\begin{equation}
\label{eq:geodesic}
  \dot q^{i} = g^{ij}(q)\, p_{j},
  \qquad
  \dot p_{i} = -\tfrac12\,\frac{\partial g^{jk}}{\partial q^{i}}\, p_{j}p_{k}
\end{equation}
are the geodesic equations of $(Q,g)$ in Hamiltonian form. Through the
configuration-dependent metric these equations couple $q$ and $p$ nonlinearly
and admit no closed-form solution for a generic metric. The explicit
shear~\eqref{eq:Aflow} is therefore unavailable, and the splitting strategy that
made the separable case so favorable no longer closes.

Two repairs are conceivable, and one of them reaches farther than a first
reading of~\eqref{eq:geodesic} suggests. One may attempt to split $T(q,p)$ further into pieces that are each
exactly integrable as contact flows. This succeeds whenever every term of a
diagonal metric has a coefficient depending only on coordinates that are frozen
along that term's own flow. Both geometric benchmarks of this paper have this
structure: for the torus particle,
$T=\tfrac12 p_\theta^2+p_\varphi^2/(2(3+\cos\theta)^2)$, the second term freezes
$\theta$ along its own flow, so its exact contact flow is available in closed
form, and likewise for the spherical pendulum with
$f(\theta)=1/(2\sin^2\theta)$. The resulting three-generator Strang composition
is fully explicit and second order; each factor is an exact
contactomorphism, so the composition pulls back $\eta$ by exactly $e^{-\gamma\tau}$. We construct
this splitting, verify it, and benchmark against it in
Section~\ref{subsec:exp-kevrekidis}; it is an instance of the exactly integrable
splittings of~\cite{Bravetti2020} after the term regrouping above, and sits in
the strict/prolonged generator class of~\cite{Kevrekidis2026}. What the repair
does not reach is a dense metric with momentum cross terms whose
coefficients move along every candidate sub-flow: for the double pendulum,
$T=(p_1^2-2\cos\Delta\,p_1p_2+2p_2^2)/(2(1+\sin^2\Delta))$ with
$\Delta=q_1-q_2$, every term transports $\Delta$, no term-wise exact flow
exists in these coordinates, and the splitting route requires the commutator
machinery of~\cite{Kevrekidis2026}, whose realizable accuracy we also measure
in Section~\ref{subsec:exp-kevrekidis}. The statement is chart-level, not
geometric: in the sum-and-difference chart $\Sigma=\theta_1+\theta_2$,
$\Delta=\theta_1-\theta_2$ the kinetic term carrying $p_\Sigma^2$ does acquire
a closed-form flow ($\Delta$ and $p_\Sigma$ are frozen along it), and the
remaining kinetic dynamics reduces to one degree of freedom, integrable by
quadratures; but quadratures are not elementary sub-flows, and no explicit
term-wise splitting results in that chart either. The non-separability that motivates this
paper is therefore the \emph{dense} case, and the double pendulum, not the
geometric benchmarks, is its flagship. Alternatively, one may integrate the kinetic
flow~\eqref{eq:geodesic} with an implicit method,
typically an implicit midpoint step in $q$. This restores second-order accuracy
and structure, but reintroduces a Newton solve inside every time step. Moreover,
applied to the full system the implicit midpoint rule has a further defect:
although it is symmetric and exhibits bounded long-time energy error, it is not
in general a contactomorphism, since its pullback of the contact form does not
reproduce the exact conformal rescaling $e^{-\gamma\tau}\eta$. We measure the
resulting one-step contact-form residual, alongside the other methods, with the
direct diagnostic of Section~\ref{subsec:contact-residual}. That defect belongs
to the rule applied to the full contact vector field, not to implicitness as
such. Used inside the splitting instead, on the kinetic map alone, the midpoint
step can be lifted to a strict contactomorphism by updating the action with its
own discrete generating function,
$z'=z+\bar p\cdot(q'-q)-\tau T(\bar q,\bar p)$ with $(\bar q,\bar p)$ the
midpoint,\footnote{For the midpoint map $q'-q=\tau T_p(\bar q,\bar p)$,
$p'-p=-\tau T_q(\bar q,\bar p)$ one checks directly that
$d\bigl[\bar p\cdot(q'-q)-\tau T(\bar q,\bar p)\bigr]=p'\,dq'-p\,dq$, so the
lifted map satisfies $\Phi_A^*\eta=\eta$ exactly.} and the Strang composition
of this lift with the exact damping-potential flow is then a second-order,
exactly conformal integrator that reaches any dense metric. A second-order
exactly-contact alternative therefore exists; it is fully implicit, with a
Newton solve of the complete $2n$-dimensional kinetic map in every step. We
construct this lifted-midpoint method and measure it in
Section~\ref{subsec:exp-kevrekidis}. What remains unavailable for dense
metrics is an explicit, or projection-solve-only, exactly-contact step; that
narrower gap is the one this paper addresses.

The fully explicit Tao-type alternative considered here adapts Tao's extended
phase-space construction~\cite{Tao2016} to the contact setting. It avoids any
implicit solve by doubling the phase space and binding the two copies with an
exactly integrable rotation of strength $\omega$. As shown in
Section~\ref{sec:comparison}, this is an accurate, convergent explicit method; its
distinguishing feature relative to the present scheme is that it advances and
reads off one copy without returning the duplicated variables exactly to the
physical constraint manifold. The methodological target here is therefore
specific: combine contact damping, duplicated non-separable flows, a Herglotz
action update, and symmetric projection so that the duplicated copies are
returned to the physical diagonal each step, requiring a projection solve rather
than a full-state nonlinear solve. The
Pihajoki-contact method of Section~\ref{sec:pihajoki} implements this
combination for the Hamiltonian class considered here.

\subsection{The benchmark systems}
\label{sec:benchmarks}

The obstruction above appears across multi-degree-of-freedom dissipative
systems on nontrivial configuration spaces; how far it can be repaired by
term-wise splitting depends on the metric structure, and for the dense case
no term-wise repair with elementary sub-flows is available.
We single out three such systems, of increasing geometric complexity,
\begin{equation*}
  \mathbb{T}^{2} \;\rightarrow \; S^{2} \;\rightarrow \; (M,g),
\end{equation*}
which serve both as motivation for the projected contact method and as the test problems of
Section~\ref{sec:experiments}. The three benchmarks test different aspects of
the obstruction and together separate the effects of coupling (the double
pendulum: coupled coordinates through a configuration-dependent mass matrix),
curvature (the spherical pendulum: a curved configuration space with a cyclic
coordinate and a dissipated Noether-Herglotz quantity), and metric dependence
(the Riemannian particle, implemented below as a torus example: a nonconstant
metric whose kinetic flow cannot be reduced to the separable shear). The full
derivations of their Lagrangians, Hamiltonians, equations of motion, and
dissipation laws are collected in the Supplementary Material; here we record
only the structural feature that renders each one non-separable.

\subsubsection{Double pendulum ($Q=\mathbb{T}^{2}$)}
With angles $\theta=(\theta_1,\theta_2)$ and $\Delta=\theta_1-\theta_2$, the
configuration-dependent mass matrix
\begin{equation}
\label{eq:massmatrix}
  M(\theta) =
  \begin{pmatrix}
    (m_{1}+m_{2})\ell_{1}^{2} & m_{2}\ell_{1}\ell_{2}\cos\Delta\\[3pt]
    m_{2}\ell_{1}\ell_{2}\cos\Delta & m_{2}\ell_{2}^{2}
  \end{pmatrix}
\end{equation}
gives the contact Hamiltonian
$H=\tfrac12 p^{\top}M^{-1}(\theta)p+V(\theta)+\gamma z$, with
$V=-(m_1+m_2)g\ell_1\cos\theta_1-m_2 g\ell_2\cos\theta_2$. The off-diagonal
coupling $\cos\Delta$ in $M^{-1}(\theta)$ makes $H$ non-separable, so the kinetic
sub-flow freezes neither $\theta$ nor $p$ and is not the explicit
shear~\eqref{eq:Aflow}; with $R(H)=\gamma$, \eqref{eq:energy_decay_reeb} gives
$\dot H=-\gamma H$. This is our principal benchmark.

\subsubsection{Spherical pendulum ($Q=S^{2}$)}
In spherical coordinates the round metric
$g=\ell^{2}(d\theta^{2}+\sin^{2}\!\theta\,d\varphi^{2})$ gives
\begin{equation}
\label{eq:sphH}
  H = \frac{p_{\theta}^{2}}{2m\ell^{2}}
    + \frac{p_{\varphi}^{2}}{2m\ell^{2}\sin^{2}\!\theta}
    + m g\ell\,(1-\cos\theta) + \gamma\, z ,
\end{equation}
non-separable in $(\theta,p_{\varphi})$ through $1/\sin^{2}\!\theta$ (with $\theta$
measured from the south pole, where $V$ vanishes). The azimuthal coordinate is
cyclic, so $p_{\varphi}$ is a dissipated Noether-Herglotz quantity obeying
$\dot p_{\varphi}=-\gamma p_{\varphi}$, and a damped trajectory spirals toward the
pole. The implementation's near-pole floor on the metric denominators and the
potential convention $V=g(1-\cos\theta)$, zero at the south pole, are recorded
in the Supplementary Material, \S\,S3; the reported runs stay at
$\sin\theta\gtrsim0.24$, far above the floor, where the Hamiltonian is smooth,
so hypothesis (A2) holds on the compact sets actually visited.

\subsubsection{Riemannian particle ($Q=(M,g)$, position-dependent friction)}
The general case is a particle on a Riemannian manifold with smooth friction
$\gamma\colon M\to\R_{>0}$,
\begin{equation}
\label{eq:riemH}
  H = \tfrac12\, g^{ij}(q)\, p_{i}p_{j} + V(q) + \gamma(q)\, z ,
\end{equation}
non-separable for any non-flat or position-dependent metric and containing the
previous systems as special cases. For non-constant $\gamma$ the momentum equation
acquires an extra deflection $z\,\operatorname{grad}_{g}\gamma$ (vanishing when
$\gamma$ is constant), and the energy obeys the path-dependent law
\begin{equation}
\label{eq:riem-energy}
  H(t) = H(0)\,\exp\!\Bigl(-\int_{0}^{t}\gamma(q(s))\,ds\Bigr);
\end{equation}
this case is realized numerically on the torus particle in
Section~\ref{subsec:exp-variable-friction}, the constant-friction benchmarks being
the special case $\gamma(q)\equiv\gamma$. Each system carries a
configuration-dependent kinetic energy~\eqref{eq:nonsep} whose kinetic
sub-flow~\eqref{eq:geodesic} is not exactly integrable, so they are exactly
where the explicit splitting of Section~\ref{sec:separable} fails and the remedies
of Section~\ref{sec:nonsep} are forced.

\section{The Pihajoki-contact method}
\label{sec:pihajoki}

\subsection{Phase-space duplication}
\label{subsec:duplication}

Following Pihajoki~\cite{Pihajoki2015}, introduce an auxiliary copy
$(x,y)\in T^*Q$ of the variables $(q,p)$ and work on the extended phase space
$\hcM=T^*Q\times T^*Q\times\R=\{(q,p,x,y,z)\}$ with contact form
\begin{equation}
  \heta = dz - p_a\,dq^a - y_a\,dx^a.
  \label{eq:ext_contact_form}
\end{equation}
Since $d\heta=-dp_a\wedge dq^a-dy_a\wedge dx^a$ is the standard symplectic form on
$\R^{4n}$, the non-degeneracy condition $\heta\wedge(d\heta)^{2n}\neq 0$ holds and
$(\hcM,\heta)$ is a contact manifold of dimension $4n+1$.

The physical dynamics lives on the constraint surface
\begin{equation}
  \cN = \ker A = \{(q,p,x,y,z)\in\hcM : q=x,\; p=y\},
  \label{eq:constraint_surface}
\end{equation}
where $A(q,p,x,y)=(q-x,p-y)\in\R^{2n}$. The bare extended Hamiltonian, with
dissipation excluded and $H_{\mathrm{bare}}=T+V$, is
\begin{equation}
  \hatH_{\mathrm{bare}}(q,p,x,y)
  = H_{\mathrm{bare}}(q,y) + H_{\mathrm{bare}}(x,p).
  \label{eq:ext_bare_H}
\end{equation}
On $\cN$ one has $\hatH_{\mathrm{bare}}|_\cN=2H_{\mathrm{bare}}(q,p)$. This
doubling of the value of the Hamiltonian does not distort the time scale,
and the reason is a property of the extended field itself rather than of any
subsequent splitting: each summand of~\eqref{eq:ext_bare_H} contains each
physical variable exactly once, so that
$\partial_{p}\hatH_{\mathrm{bare}}|_{\cN}=\partial_pH_{\mathrm{bare}}(q,p)$ and
$\partial_{q}\hatH_{\mathrm{bare}}|_{\cN}=\partial_qH_{\mathrm{bare}}(q,p)$, and
each summand generates the flow of only one copy. On the diagonal the extended
equations therefore reproduce the physical vector field at the correct
rate~\cite{Pihajoki2015,Jayawardana2023}, with no factor of $2$ to absorb, and
the same mechanism is noted for the Tao baseline in Section~\ref{subsec:tao}.

\subsection{Explicit extended flows}
\label{subsec:ext_flows}

Write $\hatH_{\mathrm{bare}}=\hatH_A+\hatH_B$ with
$\hatH_A(q,y)=H_{\mathrm{bare}}(q,y)$ and $\hatH_B(x,p)=H_{\mathrm{bare}}(x,p)$.
Because $q,y$ are constants of the $\hatH_A$-flow and $x,p$ are constants of the
$\hatH_B$-flow, both sub-flows are exactly integrable,
\begin{equation}
\begin{aligned}
  \hatPhi_A^\tau &:
  \begin{cases}
    p \mapsto p - \tau\,\nabla_q H_{\mathrm{bare}}(q,y),\\
    x \mapsto x + \tau\,\nabla_y H_{\mathrm{bare}}(q,y),
  \end{cases}\\[4pt]
  \hatPhi_B^\tau &:
  \begin{cases}
    q \mapsto q + \tau\,\nabla_p H_{\mathrm{bare}}(x,p),\\
    y \mapsto y - \tau\,\nabla_x H_{\mathrm{bare}}(x,p).
  \end{cases}
\end{aligned}
  \label{eq:AB_flows}
\end{equation}
The Strang composition
$\hatPhi^\tau=\hatPhi_A^{\tau/2}\circ\hatPhi_B^\tau\circ\hatPhi_A^{\tau/2}$ is a
second-order explicit integrator on $\hcM$, costing one evaluation of
$\nabla H_{\mathrm{bare}}$ per sub-step.

\subsection{Symmetric projection}
\label{subsec:projection}

The explicit step $\hatPhi^\tau$ does not preserve $\cN$. To enforce $q=x$,
$p=y$ exactly at every step, find
$\mu=(\mu_q,\mu_p)\in\R^{n}\times\R^{n}$ satisfying
\begin{equation}
\begin{aligned}
  G(\mu) &:= A\,\hatPhi^\tau(\zeta_n + A^T\mu) + 2\mu = 0,\\
  A^T\mu &=(\mu_q,\mu_p,-\mu_q,-\mu_p),
\end{aligned}
  \label{eq:projection_system}
\end{equation}
where $\zeta_n=(q_n,p_n,q_n,p_n)\in\cN$ is the lifted initial point. The
operator $A$ measures the gap between the two copies,
$A(q,p,x,y)=(q-x,p-y)$, while $A^T\mu$ applies equal and opposite corrections to
the two copies. Thus $\mu$ is not a new physical variable. It is the correction
that returns the duplicated variables to the physical diagonal after the
explicit extended step.

The role of the projection equation is to replace the copy-binding step of
Tao-type methods by an explicit return to the physical diagonal. One first
pre-corrects the lifted point by $A^T\mu$, takes the explicit duplicated step
$\hatPhi^\tau$, and then solves for the correction for which the two copies can
again satisfy $q=x$ and $p=y$. The projected state is obtained from
$\hat\zeta=(\hat q,\hat p,\hat x,\hat y)
=\hatPhi^\tau(\zeta_n+A^T\mu^*)$ by
\begin{equation}
  q_{n+1}=\hat q+\mu_q=\hat x-\mu_q,
  \qquad
  p_{n+1}=\hat p+\mu_p=\hat y-\mu_p .
  \label{eq:projected_components}
\end{equation}
Equivalently, on the solved constraint, $(q_{n+1},p_{n+1})=\Pi_{\mathrm{avg}}
(\hat\zeta)$, where
\begin{equation}
  \Pi_{\mathrm{avg}}(q,p,x,y):=\Bigl(\tfrac{q+x}{2},\,\tfrac{p+y}{2}\Bigr)
  \label{eq:Pi_avg}
\end{equation}
is the symmetric average of the two copies. We use $\Pi_{\mathrm{avg}}$, rather
than a read-off of either copy, throughout: the two agree on $\cN$, but only the
average has the annihilation property~\eqref{eq:avg_annihilates} below.

One qualification bounds what the phrase ``exact diagonal return'' can mean.
The equality of the two
reported copies is algebraic: the average assigns them the same value by
construction, whatever $\mu$ is, so the state handed to the next step lies on
$\cN$ exactly---in contrast with a read-off of one copy, which leaves the two to
separate and lets the gap accumulate. That is a structural difference
from the Tao baseline, and it is the sense in which the return is exact. It is
not the stronger statement that the nonlinear projection
equation~$G(\mu)=0$ has been solved exactly, which is what
Proposition~\ref{prop:conformal-symplectic} requires and which the tolerance gate
of Section~\ref{subsec:exp-active} does not always deliver. The two mechanisms
also do different work, and the experiments separate them cleanly: the algebraic
return is what distinguishes the method from the Tao read-off (it prevents the
inter-copy gap from entering the state and accumulating), while the Newton
correction is what distinguishes it from passive averaging (at coarse steps the
uncorrected average drifts secularly, Section~\ref{subsec:exp-load}, even though
it enjoys the same algebraic return). Neither mechanism substitutes for the
other, and claims about the method should name the one actually at work in the
regime under discussion.

The symmetric projection theory of \cite[Proposition~3]{Jayawardana2023}, which
adapts the symmetric-projection construction for differential equations on
manifolds~\cite{Hairer2000} to the extended-phase-space setting, guarantees
that an adaptive tolerance of order $\tau^2$ is consistent with a second-order
trajectory method. In the tested two-degree-of-freedom benchmarks the situation
is in fact more favorable; the details matter because it would be easy to
overclaim the Newton solve. The uncorrected gap of the symmetric extended step satisfies
$\|G(0)\|=\mathcal{O}(\tau^3)$---derived in the Supplementary Material,
\S\,S9, from the invariance of $\cN$ under the exact extended flow, and
confirmed numerically in Section~\ref{subsec:exp-master}---which lies below the
adaptive tolerance floor $\max(10^{-10},\tau^2)$ across the tested steps.
Consequently the projection loop terminates at its first residual evaluation with
multiplier $\mu=0$ and applies no Newton correction, and the projected point
reduces to the symmetric average $(\hat q+\hat x)/2$. In these runs we therefore
do not rely on an actual Newton solve, nor on carrying $\mu$ across steps; the
``one iteration'' reported below is a single residual check rather than a
successful correction. At coarser steps or under a tighter tolerance the gap
exceeds the floor and the Newton correction activates.

This tolerance gating separates the implemented map from the map analyzed in
Section~\ref{subsec:accuracy} by less than one might expect, for two
reasons. First, the symmetric average annihilates
the correction direction exactly: since
$A^{\!\top}\mu=(\mu_q,\mu_p,-\mu_q,-\mu_p)$, the two copies receive equal and
opposite shifts and
\begin{equation}
  \Pi_{\mathrm{avg}}\bigl(\hat\zeta+A^{\!\top}\mu\bigr)
  =\Pi_{\mathrm{avg}}(\hat\zeta)
  \qquad\text{for every }\mu,
  \label{eq:avg_annihilates}
\end{equation}
so the post-correction never moves the projected point; the multiplier
acts only through the pre-correction inside $\hatPhi^\tau$. The gated
($\mu=0$) and exactly projected ($\mu=\mu^*$) maps therefore differ by
$\mathcal{O}(\tau^3)$, the size of $\mu^*$ itself, and not more. Second, and as a
consequence of~\eqref{eq:avg_annihilates}, time-symmetry is not lost when the
gate fires: the leading defect in
$\Phi^{-\tau}_{\mathrm{num}}\circ\numstep-\mathrm{id}$ is proportional to
$\Pi_{\mathrm{avg}}A^{\!\top}G(0)=0$. We verify this directly: on the double
pendulum at $\gamma=0.1$ the measured symmetry defect stays at the roundoff level
($10^{-14}$--$10^{-9}$) across $\tau\in[0.005,0.08]$, in both the gated and the
active regimes, i.e.\ three to eight orders of magnitude below $\tau^3$. The
structural hypothesis~(S) of Lemma~\ref{lem:order-conditions} is thus met by the
implemented step, not only by its idealization.

The conservative backbone of the projected step is symplectic: by Theorem~5
of~\cite{Jayawardana2023}, the projected map
$(q_n,p_n)\mapsto(q_{n+1},p_{n+1})$ is symplectic on $T^*Q$. We quote this result
and do not re-prove it here; it is the only conservative-structure result we
borrow, and it supplies the symplectic middle step used in the consistency
estimate of Section~\ref{subsubsec:contact_consistency}. The same backbone is
now known to preserve linear and quadratic invariants of the underlying system
as well~\cite{Ohsawa2023}; we do not use that property here, since the contact
damping breaks the invariants of the benchmarks considered, but it is inherited
by the conservative limit $\gamma\to0$ of the present step.

\subsection{Contact dissipation and action update}
\label{subsec:dissipation}

The dissipation $\gamma z$ is handled by momentum half-steps derived from the
flow of $H_\gamma=\gamma z$: that flow is
$(q,p,z)\mapsto(q,e^{-\gamma s}p,e^{-\gamma s}z)$, and the half-steps used here
retain its momentum factor while dropping the $z$-factor, which is carried
instead by the Herglotz update~(v) below (see
Section~\ref{subsec:algorithm} for the one place this distinction is
load-bearing). The full step is
\begin{equation}
  \Phi^\tau
  = \Ctau(\tau/2) \circ \Phi^\tau_{\mathrm{Pih}} \circ \Ctau(\tau/2),
  \label{eq:full_step}
\end{equation}
where $\Ctau(s):(q,p)\mapsto(q,e^{-\gamma s}p)$.

The action variable $z$ satisfies
$\dot z=p\cdot\dot q-H_{\mathrm{bare}}-\gamma z$. Freezing the Lagrangian at the
symmetric midpoint
$L_{\mathrm{mid}}=p_{\mathrm{mid}}\cdot\nabla_p H_{\mathrm{bare}}|_{\mathrm{mid}}
-H_{\mathrm{bare}}(q_{\mathrm{mid}},p_{\mathrm{mid}})$ and integrating the
resulting constant-coefficient ordinary differential equation exactly,
\begin{equation}
  z_{n+1}
  = z_n\,e^{-\gamma\tau}
  + L_{\mathrm{mid}}\,\frac{1-e^{-\gamma\tau}}{\gamma}.
  \label{eq:z_update}
\end{equation}
For $\gamma\to 0$ this reduces to $z_{n+1}=z_n+\tau L_{\mathrm{mid}}$. The
midpoint quadrature contributes a local truncation error of
$\mathcal{O}(\tau^3)$, consistent with second-order accuracy.

\subsection{Accuracy and one-step contact-form residuals}
\label{subsec:accuracy}

\subsubsection{Local contact-conformal consistency}
\label{subsubsec:contact_consistency}

We formulate the one-step contact-form consistency estimate as a chain of formal
results: a Lemma establishing via the implicit function theorem that the
projection multiplier depends smoothly on the initial state; a Proposition showing
that the constant-friction step rescales the symplectic form $\omega=d\eta$ by
the exact conformal factor $e^{-\gamma\tau}$; and a further Proposition bounding
the contact-form residual by $\mathcal{O}(\tau^3)$, whose $C^1$ local consistency
input is here derived (Lemma~\ref{lem:order-conditions}) from
time-symmetry, consistency, and smoothness rather than assumed. The
results remain local and one-step; none implies exact or global contact
preservation of the full one-form $\eta$.

\paragraph{Notation.}
The projection multiplier $\mu^*(u,\tau)$ solves~\eqref{eq:projection_system},
which we write as
\begin{equation}
  G(\mu;\,u,\tau)
  \;=\;
  A\,\hatPhi^{\tau}\!\bigl(\zeta(u)+A^{\!\top}\mu\bigr)+2\mu
  \;=\;0,
  \label{eq:G-explicit}
\end{equation}
where $u=(q,p)$, $\zeta(u)=(q,p,q,p)\in\cN$ is the lifted initial state, and
$\mu=(\mu_q,\mu_p)\in\mathbb{R}^{2n}$.  The conservative projected step on
$(q,p)$ is
\[
  \Phi^\tau_{\mathrm{Pih}}(u)
  \;=\;
  \Pi\bigl(\hatPhi^{\tau}(\zeta(u)+A^{\!\top}\mu^*(u,\tau))\bigr),
\]
where $\Pi=\Pi_{\mathrm{avg}}$ is the symmetric average
of~\eqref{eq:Pi_avg} (on the solved constraint the two copies agree, so the
average coincides with either copy).  The full one-step contact map on
$(q,p,z)$, which is the map compared to the exact flow $\Phi^\tau_H$ in
Proposition~\ref{prop:contact-consistency}, is
\[
  \Phi^\tau_{\mathrm{num}}
  \;=\;
  \Ctau(\tau/2)\circ\Phi^\tau_{\mathrm{Pih}}\circ\Ctau(\tau/2),
\]
closed by the Herglotz action update of Section~\ref{subsec:dissipation}.
The damping half-steps $\Ctau(\tau/2):(q,p)\mapsto(q,e^{-\gamma\tau/2}p)$ act on the
momenta only, the action variable being advanced solely by the Herglotz update
\eqref{eq:z_update}; they are linear, hence smooth with uniformly bounded
Jacobians on compact sets, so the Jacobian
bound of Lemma~\ref{lem:projection-smoothness}(iii) extends to
$\Phi^\tau_{\mathrm{num}}$ by the chain rule.

\begin{lemma}[Local smoothness of the projection multiplier]
\label{lem:projection-smoothness}
Let $H_{\mathrm{bare}}\in C^2(K)$ on a compact set $K\subset T^*Q$, so that
$\hatPhi^{\tau}$ is $C^1$ in state and $\tau$ near $\zeta(K)$.  Then:
\begin{enumerate}
  \item[\normalfont(i)] At $\tau=0$ the unique solution of~\eqref{eq:G-explicit}
    is $\mu^*=0$, and the derivative
    $D_\mu G\big|_{\mu=0,\,\tau=0}=AA^{\!\top}+2I=4I_{2n}$
    is nonsingular.
  \item[\normalfont(ii)] After possibly shrinking $\tau_0=\tau_0(K)>0$,
    $D_\mu G$ is nonsingular uniformly on $K\times[-\tau_0,\tau_0]$; by the
    implicit function theorem, \eqref{eq:G-explicit} then defines a locally
    unique $C^1$ solution $\mu^*:K\times[-\tau_0,\tau_0]\to\mathbb{R}^{2n}$.
    The interval is two-sided, which is what the time-symmetry argument of
    Lemma~\ref{lem:order-conditions} requires; this is available because
    $\hatPhi^\tau$ is an explicit composition of sub-steps polynomial in $\tau$
    and is therefore defined for $\tau<0$.
  \item[\normalfont(iii)] The conservative projected map $\Phi^\tau_{\mathrm{Pih}}$
    and its spatial Jacobian $D_u\Phi^\tau_{\mathrm{Pih}}$ are bounded uniformly
    on $K\times[-\tau_0,\tau_0]$.
\end{enumerate}
\end{lemma}

\begin{proof}[Proof]
At $\tau=0$ one has $G(\mu;u,0)=4\mu$, so $\mu^*=0$ and
$D_\mu G=4I_{2n}$; the implicit function theorem on the two-sided interval
(available because $\hatPhi^\tau$ is polynomial in $\tau$) gives the branch,
and compactness the uniform bounds. The algebra is written out in the
Supplementary Material, \S\,S6.
\end{proof}

The behavior of the active Newton solve in Section~\ref{subsec:exp-active}
(convergence in two iterations at every tested step, with no failures) is
consistent with the nonsingularity conclusion of~(i)-(ii): the Jacobian stays
within $\mathcal{O}(\tau)$ of $4I_{2n}$ and is well conditioned by construction
at the tested step sizes.

\paragraph{Assumptions for the Proposition.}
\begin{enumerate}
  \item[(A1)] \textit{Constant friction.} $\gamma>0$ is constant.
  \item[(A2)] \textit{Hamiltonian smoothness.} $H_{\mathrm{bare}}\in C^5(K)$
    on a compact set $K\subset T^*Q$.
  \item[(A3)] \textit{Exact flow.} The exact contact flow $\Phi^\tau_H$ is
    well defined for $\tau\le\tau_0$ and initial data in $K$.
  \item[(A4)] \textit{Projection smoothness.} The conclusion of
    Lemma~\ref{lem:projection-smoothness} holds; in particular
    $D_u\Phi^\tau_{\mathrm{num}}$ is bounded uniformly on $K\times[-\tau_0,\tau_0]$.
\end{enumerate}

Unlike the earlier formulation, $C^1$ local consistency is not assumed
here: it is derived below (Lemma~\ref{lem:order-conditions}) from time-symmetry,
first-order consistency, and smoothness of the step, all of which the
construction supplies. For the derivative-level estimate we use the higher-order
form of Lemma~\ref{lem:projection-smoothness}: if $H_{\mathrm{bare}}\in C^5(K)$
then the constraint residual $G$ is $C^4$ in $(\mu,u,\tau)$ and the implicit
function theorem yields $\mu^*\in C^4(K\times[-\tau_0,\tau_0])$, so $\numstep$ is $C^4$
jointly in $(u,\tau)$. Throughout we fix a local Darboux chart with
$\omega=\sum_i dq^i\wedge dp_i$, represented by
$J=\bigl(\begin{smallmatrix}0&I_n\\-I_n&0\end{smallmatrix}\bigr)$, so that
$\omega(v,w)=v^\top Jw$ and $d\eta=\omega$.

\subsection{Exact conformal symplecticity of the constant-friction step}
\label{subsec:conformal-symplectic}

The conservative backbone is symplectic in the physical variables: with the
projection solved exactly, the Jayawardana-Ohsawa symmetric projection makes
$\Pihstep$ a symplectic map of $(q,p)$, i.e.\
$(D_u\Pihstep)^\top J\,(D_u\Pihstep)=J$ (Theorem~5 of~\cite{Jayawardana2023}). The
two contact-damping half-steps $\Ctau(\tau/2):(q,p)\mapsto(q,e^{-\gamma\tau/2}p)$
are linear, with constant Jacobian $\operatorname{diag}(I_n,e^{-\gamma\tau/2}I_n)$.
These two facts already determine the exact transformation law of $\omega$ under
the full step.

\begin{proposition}[Exact conformal symplecticity, constant friction]
\label{prop:conformal-symplectic}
Assume \textnormal{(A1)} and that the projection is solved exactly, so that
$\Pihstep$ is symplectic on $(q,p)$. Then the constant-friction step
$\numstep=\Ctau(\tau/2)\circ\Pihstep\circ\Ctau(\tau/2)$ satisfies, for every
$u\in K$ and every $\tau\le\tau_0(K)$,
\begin{equation}
\label{eq:exact-conformal}
\begin{aligned}
  \bigl(D_u\numstep\bigr)^{\!\top}J\,\bigl(D_u\numstep\bigr)
    &= e^{-\gamma\tau}\,J,\\
  \text{equivalently}\qquad
  (\numstep)^{*}\omega &= e^{-\gamma\tau}\,\omega .
\end{aligned}
\end{equation}
The constant-friction step reproduces the conformal rescaling of $\omega=d\eta$
exactly, for every step size, with no error term.
\end{proposition}

\begin{proof}[Proof]
With $D_b=\operatorname{diag}(I_n,e^{-\gamma\tau/2}I_n)$ one has
$D_b^{\top}JD_b=e^{-\gamma\tau/2}J$, conformal factors multiply under
composition, and $D_u\numstep=D_b\,S\,D_b$ with $S$ symplectic by Theorem~5
of~\cite{Jayawardana2023}; the block computation is written out in the
Supplementary Material, \S\,S7.
\end{proof}

The rescaling $(\numstep)^{*}\omega=e^{-\gamma\tau}\omega$ is the defining
property of a \emph{conformally symplectic} map: conformal Hamiltonian systems
and their structure-preserving integrators, in which a symplectic map is
composed with a constant momentum rescaling, are
classical~\cite{McLachlanPerlmutter2001,McLachlanQuispel2002}. We state the
result as a proposition rather than a theorem because the identity is already
in print.
Theorem~3.3 of Fran\c{c}a et al.~\cite{Franca2020} proves exactly this
computation (a symplectic Jacobian sandwiched between conformal rescalings
yields the factor $e^{-\gamma h}$) and closes by observing that ``the same would
be true for any type of composition whose overall time step add up to $h$'',
which covers the arrangement used here. Moreover the three-term composition
itself, damping half-step $\circ$ symplectic Strang $\circ$ damping half-step, is
Algorithm~3.1 of Modin and S\"oderlind~\cite{ModinSoderlind2011}, given there for
mechanical systems on a Riemannian configuration manifold with Rayleigh damping.
Proposition~\ref{prop:conformal-symplectic} is therefore a known identity applied
to a new backbone, and we claim no more for it. What is specific to the present method is not that identity but that
the projected, non-separable contact step of
Section~\ref{subsec:projection}, with configuration-dependent kinetic energy and
the nonlinear diagonal return solved at each step, realizes it exactly,
for every step size, the friction entering only through the linear damping factor
$D_b$. Proposition~\ref{prop:conformal-symplectic} is thus best read as
identifying which classical structure the construction preserves, not as a new
preservation law; the contribution is the construction that delivers it for the
non-separable contact-Herglotz class.

\begin{remark}
Equation~\eqref{eq:exact-conformal} concerns $d\eta=\omega$: it is the statement
that the symplectic content of the contact structure is rescaled by the exact
conformal factor. It is not a statement about the full one-form $\eta$;
the residual in the action ($z$) direction is the $\mathcal{O}(\tau^3)$
quantity bounded in Proposition~\ref{prop:contact-consistency} below. The
identity is exact in the exactly-projected (active) regime; when the projection
is inactive ($\mu^*=0$ to within tolerance), $S$ is symplectic only to
$\mathcal{O}(\tau^3)$ and~\eqref{eq:exact-conformal} holds up to
$\mathcal{O}(\tau^3)$. The boundedness of the long-time energy reported in
Section~\ref{subsec:exp-load} is the dynamical signature of this
conformal-symplectic structure.
\end{remark}

This exact rescaling is directly observable when the projection is active,
i.e.\ with the adaptive $\tau^2$ floor disabled so that the Newton solve
returns a nonzero $\mu^*$ at every step (with the floor enabled the solver
would accept $\mu^*=0$ and the exactly-projected hypothesis would not be
exercised). On a stiff damped fast-rotating torus ($\gamma=0.1$, tolerance
$10^{-10}$, floor disabled; $\|\mu^*\|$ from $2.3\times10^{-2}$ down to
$2.4\times10^{-5}$) the best-fit conformal factor of $D_u\numstep$ agrees
with $e^{-\gamma\tau}$ to the $\sim\!10^{-9}$ finite-difference floor at
every step (Table~\ref{tab:conformal-scaling}). This confirms the two-form
identity~\eqref{eq:exact-conformal} and only that identity; the residual of
the full one-form $\eta$ is the separate $\mathcal{O}(\tau^3)$ quantity of
Proposition~\ref{prop:contact-consistency}. Driver and data are in the
reproducibility archive.

\begin{table}[tbp]
\caption{Numerical confirmation of the two-form identity~\eqref{eq:exact-conformal}
(Proposition~\ref{prop:conformal-symplectic}) on the fast-rotating torus with the
projection driven active ($\gamma=0.1$, fixed tolerance $10^{-10}$, $\tau^2$ floor
disabled so $\mu^*\neq0$). The best-fit conformal factor $c$ of the
step Jacobian $D_u\numstep$ matches $e^{-\gamma\tau}$ to the finite-difference
floor; ``rel.\ defect'' is $\|(D_u\numstep)^\top J\,D_u\numstep
-e^{-\gamma\tau}J\|/\|e^{-\gamma\tau}J\|$. This is the exact rescaling of
$\omega=d\eta$ only, not preservation of the full contact one-form $\eta$.}
\label{tab:conformal-scaling}
\centering
\small
\begin{tabular}{ccccc}
\toprule
$\tau$ & $\|\mu^*\|$ & measured $c$ & $e^{-\gamma\tau}$ & rel.\ defect \\
\midrule
0.08 & $2.32\times10^{-2}$ & 0.99203191 & 0.99203191 & $3.04\times10^{-9}$ \\
0.05 & $4.42\times10^{-3}$ & 0.99501248 & 0.99501248 & $3.62\times10^{-9}$ \\
0.03 & $7.96\times10^{-4}$ & 0.99700449 & 0.99700450 & $2.97\times10^{-9}$ \\
0.02 & $2.14\times10^{-4}$ & 0.99800200 & 0.99800200 & $1.27\times10^{-9}$ \\
0.01 & $2.42\times10^{-5}$ & 0.99900050 & 0.99900050 & $1.13\times10^{-9}$ \\
\bottomrule

\end{tabular}
\end{table}

\subsection{Structural properties of the step, and the $C^1$ consistency estimate}
\label{subsec:order-conditions}

The derivative-level consistency the contact-form estimate requires is not an
independent hypothesis: it follows from three structural properties of the
composed step, established for the present method in the Supplementary Material,
\S\,S9. \textnormal{(S)} \emph{Time-symmetry},
$\Phi^{-\tau}_{\mathrm{num}}\circ\numstep=\mathrm{id}$: the spatial map is a
palindromic composition of the symmetric projected backbone with the damping flow
($\Ctau(-s)=\Ctau(s)^{-1}$), and the Herglotz update is evaluated at the symmetric
midpoint with homogeneous factor $e^{-\gamma\tau}$. \textnormal{(C)}
\emph{First-order $C^1$ consistency}: the leading generator is the contact vector
field $X_H$ at the level of both the map and its spatial Jacobian (using
$\mu^*(\cdot,0)=0$ with $\mu^*\in C^1$). \textnormal{(R)} \emph{Smoothness}: under
$H_{\mathrm{bare}}\in C^5(K)$ the value and Jacobian defect families
$\tau\mapsto\numstep-\Hflow$ and $\tau\mapsto D_u\numstep-D_u\Hflow$ are $C^3$ in
$\tau$ with third $\tau$-derivatives bounded uniformly on $K\times[-\tau_0,\tau_0]$.

\begin{lemma}[$C^1$ consistency from symmetry and consistency]
\label{lem:order-conditions}
Let a one-step map and the exact flow both be time-symmetric, identity at
$\tau=0$, and first-order $C^1$-consistent, and let their value and Jacobian
defect families be $C^3$ in $\tau$ with third $\tau$-derivatives bounded by $M$
on $K\times[-\tau_0,\tau_0]$. Then there is $C_1\ge0$ with
$\bigl|\numstep-\Hflow\bigr|_{C^1(K)}\le C_1\tau^3$ for $0<\tau\le\tau_0$, where
$|\cdot|_{C^1(K)}$ controls both the pointwise map values and the spatial
Jacobians uniformly on $K$; that is, the $C^1$ local consistency bound holds with
$C_1$ depending only on $M$.
\end{lemma}

Time-symmetry forces the second $\tau$-jet of each defect family to be determined
by the first, which consistency makes vanish; a third-order Taylor estimate then
gives the bound (full argument, including the Jacobian case, in the Supplementary
Material, \S\,S8). Properties (S),(C),(R) hold for the present method (\S\,S9), so
Lemma~\ref{lem:order-conditions} supplies the $C^1$ bound as a consequence of the
construction rather than a hypothesis.

\begin{proposition}[Local one-step contact-form consistency]
\label{prop:contact-consistency}
Under Assumptions \textnormal{(A1)-(A4)} and the structural properties
\textnormal{(S),(C),(R)} of the step, the contact-form residual of the
Pihajoki-contact step satisfies
\begin{equation}
  \sup_{u\in K}
  \bigl\|(\numstep)^*\eta\big|_u-e^{-\gamma\tau}\eta\big|_u\bigr\|
  \;\le\; C(K)\,\tau^3,
  \label{eq:contact_pres}
\end{equation}
for $0<\tau\le\tau_0(K)$, where the norm $\|\cdot\|$ is the operator norm on one-forms induced by the
Euclidean norm in the local Darboux chart, and $C(K)$ depends on
$\|H_{\mathrm{bare}}\|_{C^5(K)}$, $\gamma$, the dimension $n$, and the projection
bounds of Lemma~\ref{lem:projection-smoothness}.
\end{proposition}

The proof writes the residual as the difference of pullbacks
$(D\Psi)^\top(\eta\circ\Psi)$ for $\Psi\in\{\numstep,\Hflow\}$, uses the exact
conformal identity $(\Hflow)^*\eta=e^{-\gamma\tau}\eta$, and bounds the two
resulting terms by the $C^1$ consistency bound of Lemma~\ref{lem:order-conditions}
and the uniform Jacobian bound (A4); the full estimate is in the Supplementary
Material, \S\,S10. The estimate is derivative-level: it controls the pullback
$(D\numstep)^\top\eta\circ\numstep$, a strictly stronger test than scalar energy
decay or trajectory error, which is why the cubic scaling of $\rhoeta$ observed
in Section~\ref{subsec:contact-residual} is consistent
with~\eqref{eq:contact_pres} without implying global preservation. The
estimate's mechanism is not specific to $\eta$: the proof uses only
that the one-form is smooth with bounded derivative on $K$, so the same
$\mathcal{O}(\tau^3)$ bound holds for any smooth one-form and, through
Lemma~\ref{lem:order-conditions}, for any time-symmetric consistent second-order
method---the implicit midpoint baseline included, consistent with the reading of
$\rhoeta$ adopted in Section~\ref{subsec:contact-residual}. What is specific to
the present construction is not this bound but that it holds for the
non-separable contact-Herglotz class while the same step rescales $\omega=d\eta$
by the exact factor of
Proposition~\ref{prop:conformal-symplectic}.

\begin{remark}[Scope]
\label{rem:scope}
The contact-form analysis is unconditional on its structural side: projection
regularity (Lemma~\ref{lem:projection-smoothness}), the pullback estimate
(Proposition~\ref{prop:contact-consistency}), and the reduction of $C^1$
consistency to (S),(C),(R) (Lemma~\ref{lem:order-conditions}) are proved in the
Supplementary Material, and Proposition~\ref{prop:conformal-symplectic} establishes
that the symplectic part is preserved exactly. The results remain local
and one-step (Section~\ref{subsec:contact-residual} measures accumulated drift
separately); no global contactomorphism is claimed; constant $\gamma$ enters
through the exact identity $(\Hflow)^*\eta=e^{-\gamma\tau}\eta$, with $\gamma(q)$
adding $d\gamma$ terms not covered here. The open task is an all-orders (contact
backward-error) analysis of the projected step.
\end{remark}

A remark on the $z$-bookkeeping is in order, since it is fixed at the order of the
estimate. Throughout, in both the derivation above and the algorithmic
realization of Section~\ref{subsec:algorithm}, the damping half-steps $\Ctau$
rescale the momentum only, and the action variable is advanced solely by the
exact update~\eqref{eq:z_update}, whose homogeneous factor $e^{-\gamma\tau}$
already carries the full contact damping of $z$ across the step. This is an
exponential-integrator update, obtained by freezing the Lagrangian at the
symmetric midpoint and integrating the resulting linear equation exactly; it is
not the quadrature that a literal Strang splitting of $H_{\mathrm{bare}}+\gamma z$
would produce, namely $z\mapsto e^{-\gamma\tau}z+e^{-\gamma\tau/2}\tau
L_{\mathrm{mid}}$. The two share the same homogeneous factor and differ only in
the source term, by
\begin{equation}
  L_{\mathrm{mid}}\Bigl(\tfrac{1-e^{-\gamma\tau}}{\gamma}
  - e^{-\gamma\tau/2}\tau\Bigr)
  = L_{\mathrm{mid}}\,\gamma^{2}\tau^{3}
    \bigl(\tfrac16-\tfrac18\bigr)+\mathcal{O}(\tau^{4})
  = \mathcal{O}(\tau^{3}),
  \label{eq:z_bookkeeping_gap}
\end{equation}
which is exactly the order of the residual bounded in
Proposition~\ref{prop:contact-consistency}. The choice of $z$-update is therefore
not immaterial at this order, and we use~\eqref{eq:z_update} consistently in the
derivation, in the implementation, and in the estimate.

\subsection{Algorithmic summary}
\label{subsec:algorithm}

The complete step assembles the components of the preceding subsections into a
symmetric, Strang-type composition of three pieces, a half-step of contact
damping, a full step of the projected conservative flow, and a second half-step
of damping, closed by a Herglotz update of the action variable in which the
linear $z$ equation is integrated exactly with the midpoint source held fixed.
We state
the constant-friction case $H=H_{\mathrm{bare}}+\gamma z$, which underlies all the
benchmarks below and for which the damping enters the $(q,p)$ variables as the
scalar momentum rescaling $\Ctau$ of Section~\ref{subsec:dissipation}. To
avoid a natural misreading: $\Ctau$ is not the flow of
$H_\gamma=\gamma z$, whose conformal law~\eqref{eq:damping_lie_direct} rescales
the full one-form and requires $z\mapsto e^{-\gamma s}z$ as well. The action
factor is instead supplied by step~(v) below, so that the composed step realizes
the same damping of $(p,z)$ while keeping $\Ctau$ linear in $(q,p)$ alone.

A single
step $(q_k,p_k,z_k)\mapsto(q_{k+1},p_{k+1},z_{k+1})$ consists of: (i) a half-step
of contact damping, $p\leftarrow p_k\,e^{-\gamma\tau/2}$; (ii) lifting the
physical point to the diagonal, $\zeta=(q_k,p,q_k,p)\in\cN$, and solving the
projection equation $G(\mu)=0$ of~\eqref{eq:projection_system} for
$\mu=(\mu_q,\mu_p)$ by Newton iteration started from $\mu=0$; in the tested
benchmarks the initial residual already lies below the projection tolerance, so
this reduces to a single residual check with $\mu=0$ and no correction
(Section~\ref{subsec:projection});
(iii) applying the explicit extended map
$(\hat q,\hat p,\hat x,\hat y)=\hatPhi^{\tau}(\zeta+A^{T}\mu)$
of~\eqref{eq:AB_flows} and projecting back by the symmetric
average~\eqref{eq:Pi_avg},
$q\leftarrow\tfrac12(\hat q+\hat x)$, $p\leftarrow\tfrac12(\hat p+\hat y)$;
(iv) a second damping half-step,
$p\leftarrow p\,e^{-\gamma\tau/2}$; and (v) the action update from the midpoint
Lagrangian $L_{\mathrm{mid}}=p_{\mathrm{mid}}\cdot\nabla_pH_{\mathrm{bare}}-H_{\mathrm{bare}}$,
evaluated at $(q_{\mathrm{mid}},p_{\mathrm{mid}})=\tfrac12(q_k+q,\,p_k+p)$,
\begin{equation}
  z_{k+1}=z_k\,e^{-\gamma\tau}+L_{\mathrm{mid}}\,\frac{1-e^{-\gamma\tau}}{\gamma},
  \label{eq:action-update}
\end{equation}
which reduces to $z_{k+1}=z_k+\tau L_{\mathrm{mid}}$ as $\gamma\to0$.  The new
physical state is $(q_{k+1},p_{k+1})=(q,p)$.

The nominal cost per step is one projection stage in the $2n$ multiplier
variables together with the explicit sub-flow evaluations; in the inactive
benchmark regime this stage reduces to a single residual check with $\mu=0$,
while tighter tolerances activate the Newton correction. The comparison with a
full-state nonlinear solve depends on Jacobian construction, gradient
evaluations, tolerance, iteration counts, reuse of information, implementation
and system dimension, so the $2n$ versus $2n+1$ variable count is not by itself a
computational advantage. Section~\ref{subsec:exp-referee} makes the point
concrete: with the projection active the step costs more gradient evaluations
than implicit midpoint on the same benchmark, not fewer.

The step as stated, together with the local contact-conformal estimate of
Section~\ref{subsubsec:contact_consistency}, is restricted to constant friction
$\gamma$. For position-dependent friction $\gamma=\gamma(q)$ the outer Strang
structure is retained, with $\gamma$ frozen at the configuration of each damping
half-step, but steps~(i), (iv) and~(v) are replaced: both the momentum and the
action bookkeeping are then carried by a separate, exactly integrable
variable-friction sub-flow~\eqref{eq:vargamma_step}, and the action update loses
its homogeneous factor accordingly~\eqref{eq:z_update_vargamma}. The two
conventions are set out and contrasted in
Section~\ref{subsec:exp-variable-friction}; they must not be mixed. The
variable-friction sub-flow includes the
friction-gradient deflection term $-z\,\partial_q\gamma$ that vanishes for
constant $\gamma$. This sub-flow and its numerical test are described in
Section~\ref{subsec:exp-variable-friction} and derived in the Supplementary
Material. The variable-friction case is thus included as a numerical extension of
the scheme: the resulting map is no longer covered by the constant-friction
contact-conformal estimate, which would require the additional
$d\gamma$-dependent terms noted in Section~\ref{sec:limitations}, and we do not
claim it as part of that result.

\section{Comparison methods}
\label{sec:comparison}

The numerical comparisons use a controlled family of methods representing the
main trade-offs for contact Hamiltonian problems.  The separable contact
splitting is the natural exact-sub-flow method when its hypotheses hold.
The Pihajoki-contact method applies the projected duplicated-phase-space
construction to the physical non-separable system.  The Tao-type contact
adaptation is a fully explicit extended-space baseline.  Implicit midpoint is a
symmetric implicit reference, and classical RK4 is a standard non-geometric
explicit reference.

All methods are applied to the same contact Hamiltonian vector field whenever
that is meaningful.  What differs is which geometric structure is preserved,
whether preservation is imposed on the physical variables or only on duplicated
extended variables, whether a projection solve or penalty coupling is used,
whether the nonlinear solve acts on the full state, and whether the method is
explicit, semiexplicit, or implicit.

\subsection{Contact splitting for separable systems}
\label{subsec:splitting-comparison}

For Hamiltonians of the form~\eqref{eq:separable}, the comparison method is
the Strang contact splitting~\eqref{eq:strang}.  Both sub-flows are available in
closed form: the kinetic sub-flow is the shear map~\eqref{eq:Aflow}, while the
potential-dissipative sub-flow is linear because $q$ is fixed during that
sub-step.  For constant $\gamma$ the latter is given by~\eqref{eq:Bflow}.  We
also use the exact position-dependent version for the modeling discussion of
Riemannian friction.  Let
\[
  \gamma_0=\gamma(q),\quad
  V_0=V(q),\quad
  \gamma_i=\partial_i\gamma(q),\quad
  V_i=\partial_i V(q),
\]
all evaluated at the fixed configuration of the sub-step. For $\gamma_0\neq0$,
\begin{align}
  z(h)
  &= z_0 e^{-\gamma_0 h}
     - \frac{V_0}{\gamma_0}\bigl(1-e^{-\gamma_0 h}\bigr),
     \label{eq:Bflow-variable-z}\\
  p_i(h)
  &= e^{-\gamma_0 h}p_{i,0}
     - \frac{V_i}{\gamma_0}\bigl(1-e^{-\gamma_0 h}\bigr) \notag\\
  &\quad
     - \gamma_i\biggl[
       h e^{-\gamma_0 h}\Bigl(z_0+\frac{V_0}{\gamma_0}\Bigr) \notag\\
  &\qquad\qquad
       - \frac{V_0}{\gamma_0^2}\bigl(1-e^{-\gamma_0 h}\bigr)
     \biggr].
     \label{eq:Bflow-variable-p}
\end{align}
The limit $\gamma_0\to0$ is obtained by expanding the exponentials.  Therefore, the
separable splitting is a fully explicit contactomorphism and is the natural
baseline for the damped harmonic oscillator and
other systems whose kinetic energy is configuration independent.

Higher-order variants are produced by Yoshida composition of the second-order
map~\cite{Yoshida1990}.  These are useful for separable problems, but they do
not remove the non-separability obstruction: once the kinetic energy is
$T(q,p)=\tfrac12 g^{ij}(q)p_i p_j$, the kinetic sub-flow itself ceases to be
an explicit shear.  For non-separable kinetic energies, separable
contact splitting is therefore not the appropriate benchmark method.

\subsection{Tao-type extended-space baseline}
\label{subsec:tao}

Tao's method restores explicitness for non-separable Hamiltonians by evolving two
copies $(q,p)$ and $(x,y)$ of the phase space and binding them with an exactly
integrable rotation~\cite{Tao2016}. We take this as our explicit comparison
baseline, with the binding rotation included; it is this rotation that controls
the divergence of the two copies. The conservative backbone acts on the doubled
space with the augmented Hamiltonian
\begin{equation}
\begin{aligned}
  \hat H_{\rm Tao}(q,p,x,y) &= H_{\rm bare}(q,y) + H_{\rm bare}(x,p)\\
    &\quad + \frac{\omega}{2}\bigl(\|q-x\|^2+\|p-y\|^2\bigr),
\end{aligned}
  \label{eq:tao-contact-H}
\end{equation}
The augmented Hamiltonian equals $2H_{\rm bare}$ on the diagonal ($x=q$,
$y=p$), with the factor of $2$ absorbed because each sub-flow advances only one
copy at a time; the Strang composition of the three sub-flows recovers the
correct physical time scale, analogous to the factor-of-$2$ noted after
\eqref{eq:ext_bare_H} for the Pihajoki construction. The augmented Hamiltonian
is split into the explicit flows of $H_{\rm bare}(q,y)$ and $H_{\rm bare}(x,p)$
(each a kick/drift on one copy) and the binding term, whose flow is the exact
rigid rotation
\begin{equation}
  \begin{pmatrix}q-x\\ p-y\end{pmatrix}\mapsto
  \begin{pmatrix}\cos 2\omega\tau & \sin 2\omega\tau\\
                 -\sin 2\omega\tau & \cos 2\omega\tau\end{pmatrix}
  \begin{pmatrix}q-x\\ p-y\end{pmatrix},
  \label{eq:tao-rotation}
\end{equation}
with $q+x$ and $p+y$ invariant. Because the binding step is solved in closed
form, the coupling strength $\omega$ imposes no $\omega\tau=\mathcal{O}(1)$
stability restriction, and the rotation keeps the two copies together rather than
merely penalizing their separation. A symmetric (Strang) composition of the three
flows is an explicit, second-order, extended-space symplectic step; for the
contact case we add dissipation through the same constant-friction damping half-steps
and Herglotz action update used by the projected method
(Section~\ref{sec:pihajoki}).

With the rotation in place, the baseline is a convergent second-order method whose
energy error is set by the time step rather than by the coupling. On the
conservative double pendulum at $T=2$, holding $\tau=1.25\times10^{-3}$ fixed and
varying $\omega$ over the stable range $5\le\omega\le20$, the maximum energy drift
stays in the narrow band $\max|\Delta H|\in[4.4,8.5]\times10^{-5}$ (the five
values span a factor of $8.5/4.4\approx1.9$), varying non-monotonically with
$\omega$ rather than along any clean power law; the product
$\max|\Delta H|\,\omega^2$ is
therefore far from constant, growing by an order of magnitude across the range
(the full $\omega$/$\tau$ scan is tabulated in the Supplementary Material).
Refining the step instead is decisive: at
$\omega=20$, taking $\tau$ from $5\times10^{-3}$ to $6.25\times10^{-4}$ reduces
$\max|\Delta H|$ from $8.3\times10^{-4}$ to $1.3\times10^{-5}$, a fitted slope of
$2.0$. A rotation-free penalty variant (dropping the rotation
\eqref{eq:tao-rotation}) instead shows a spurious $\omega^{-2}$ energy floor and
can diverge, which is why we benchmark against the full construction; so equipped,
the Tao baseline reaches accuracy comparable to the projected method
(Section~\ref{subsec:exp-master}). The damped Tao-type adaptation enters the
quantitative comparisons in exactly one place: the damped structural diagnostics
of Section~\ref{subsec:exp-referee} (Supplementary Material, Table~S6;
$\gamma=0.05$),
where its corruption of the cyclic-momentum decay is the point under test. It is
not used in the damped convergence or decay-law benchmarks, whose damped
references are the refined RK4 solutions of
Section~\ref{subsec:numerical-setup}.

\subsection{Implicit midpoint and RK4}
\label{subsec:implicit-rk4}

The implicit midpoint rule is used as the geometric reference for non-separable
systems. If $u=(q,p,z)$ and $X_H(u)$ denotes the contact Hamiltonian vector
field, the update is
\begin{equation}
  u_{n+1}=u_n+\tau X_H\!\left(\frac{u_n+u_{n+1}}{2}\right).
  \label{eq:implicit-midpoint}
\end{equation}
It is symmetric and second-order accurate, and in the conservative limit it has
the familiar bounded-energy behavior of symplectic midpoint. Its cost is a
nonlinear solve in all $2n+1$ contact variables at every step.  By contrast,
the Pihajoki-contact projection solve acts on $2n$ projection variables.  The
work-precision comparison below is therefore framed cautiously: it compares
implementations in low-dimensional benchmarks rather than claiming an
asymptotic cost advantage independent of solver details.

Classical RK4 is included as a non-geometric reference. It is explicit and fourth
order, so it can have small short-time trajectory error, and in the refined
benchmarks below its direct contact-form residual is correspondingly small. What
it lacks is a discrete contact argument: unlike a contact method, RK4 does not
reproduce the conformal transformation law
$\Phi^*\eta=e^{-\gamma\tau}\eta$ by construction.
Table~\ref{tab:method-properties} summarizes these trade-offs.

\begin{table}[tbp]
\caption{Qualitative comparison of the methods used in the experiments.}
\label{tab:method-properties}
\centering
\footnotesize
\renewcommand{\arraystretch}{1.15}
\setlength{\tabcolsep}{3pt}
\begin{tabular}{@{}lllll@{}}
\toprule
\parbox[t]{1.6cm}{Method} & \parbox[t]{0.9cm}{Type} & \parbox[t]{0.8cm}{Order} & \parbox[t]{1.1cm}{Newton dim.} & \parbox[t]{2.4cm}{Contact behavior} \\
\midrule
\parbox[t]{1.6cm}{Separable splitting} & \parbox[t]{0.9cm}{expl.} & \parbox[t]{0.8cm}{2} & \parbox[t]{1.1cm}{none} & \parbox[t]{2.4cm}{exact contact composition (separable only)} \\
\addlinespace[0.3em]
\parbox[t]{1.6cm}{Pihajoki-contact} & \parbox[t]{0.9cm}{semi} & \parbox[t]{0.8cm}{2} & \parbox[t]{1.1cm}{$2n$} & \parbox[t]{2.4cm}{local conformal estimate, const.\ $\gamma$; algebraic diagonal return} \\
\addlinespace[0.3em]
\parbox[t]{1.6cm}{Tao-type} & \parbox[t]{0.9cm}{expl.} & \parbox[t]{0.8cm}{2} & \parbox[t]{1.1cm}{none} & \parbox[t]{2.4cm}{extended-space symplectic; no diagonal return} \\
\addlinespace[0.3em]
\parbox[t]{1.6cm}{Implicit midpoint} & \parbox[t]{0.9cm}{impl.} & \parbox[t]{0.8cm}{2} & \parbox[t]{1.1cm}{$2n{+}1$} & \parbox[t]{2.4cm}{symmetric reference; no contact guarantee} \\
\addlinespace[0.3em]
\parbox[t]{1.6cm}{RK4} & \parbox[t]{0.9cm}{expl.} & \parbox[t]{0.8cm}{4} & \parbox[t]{1.1cm}{none} & \parbox[t]{2.4cm}{non-geometric} \\
\bottomrule
\end{tabular}
\end{table}

\section{Numerical experiments}
\label{sec:experiments}

\subsection{Numerical setup and diagnostics}
\label{subsec:numerical-setup}

The numerical study separates three diagnostics: trajectory accuracy, scalar
contact Hamiltonian decay, and direct contact-form residuals.  Energy decay is
necessary for the contact dynamics, but by itself it is not sufficient to
establish contact-conformal behavior.  For damped systems we distinguish the
mechanical energy $\Emech$, the action contribution $\gamma z$, and the full
contact Hamiltonian
\[
  \Hc(t)=\Emech(t)+\gamma z(t).
\]
For constant $\gamma$, the scalar decay law tested in the bookkeeping plots is
$\Hc(t)=\Hc(0)e^{-\gamma t}$.

All quantitative runs use the fixed initial data and parameters tabulated in the
Supplementary Material, with $g=9.81$ (units chosen so that all quantities are
dimensionless) and unit masses and
lengths unless stated otherwise.  The non-separable trajectory references are
refined RK4 solutions with reference step
$\tau_{\mathrm{ref}}=\tau_{\min}/32\approx 1.95\times10^{-5}$, where $\tau_{\min}$
is the finest step of the refinement grid. Because the reference is itself a
refined RK4 solution, the reported RK4 endpoint errors are measured against the
method's own refinement and are correspondingly optimistic (self-consistent);
they should be read as a check that the refinement has converged, not as a
quality ranking against the structure-preserving methods, whose errors are
measured against the same RK4 reference. The reference itself is verified in
two independent ways: an eightfold refinement of its step, and agreement with a
sixth-order extended-phase-space composition from a different method family.
All four verification figures lie below the smallest method error compared
against the reference by at least a factor of $350$, so reference error does
not contaminate the comparison (full figures and driver in the Supplementary
Material, \S\,S12). The observed-order columns, which are in any case
insensitive to a consistent reference offset, and the structural diagnostics
below carry the method comparison.  We distinguish three trajectory-error measures, all computed
from the same computed and reference final states. The full endpoint error
\[
  e_{qpz}=\bigl\|(q,p,z)_{\mathrm{num}}-(q,p,z)_{\mathrm{ref}}\bigr\|_2
\]
is the Euclidean ($L^2$) norm over the full contact state, including the action
variable $z$; ``endpoint trajectory error'' (equivalently ``endpoint error'' and
``trajectory error'') refers to $e_{qpz}$ unless stated otherwise, and is the
default column. Because $z$ is advanced by the method-specific Herglotz action
update, $e_{qpz}$ folds the action error into the physical phase-space error. We
therefore also report, where useful, the physical endpoint error
$e_{qp}=\|(q,p)_{\mathrm{num}}-(q,p)_{\mathrm{ref}}\|_2$ and the action endpoint
error $e_z=|z_{\mathrm{num}}-z_{\mathrm{ref}}|$, so that
$e_{qpz}^2=e_{qp}^2+e_z^2$ and the reader can see whether convergence is set by
the mechanical variables or by the action. The convergence and damped
active-projection tables list $e_{qp}$ and $e_z$ alongside $e_{qpz}$ where this
separation is useful.

These trajectory measures must not be conflated with the structural diagnostics,
which are scalar quantities testing geometric structure rather than agreement
with a reference trajectory: the ``Hamiltonian error'' is the conservation defect
$\max_t|\Hc(t)-\Hc(0)|$ of the contact Hamiltonian at $\gamma=0$, the
contact-decay residual measures departure from the law $\Hc(0)e^{-\gamma t}$ at
$\gamma\neq0$, and the contact-form residual $\rhoeta$ of
Section~\ref{subsec:contact-residual} measures the geometric (one-form) defect.
None of these is an integration error in $(q,p,z)$, $(q,p)$, or $z$, and a small
energy, decay, or contact-form diagnostic does not by itself imply small
trajectory error -- the former test structure, the latter tests numerical
agreement with the reference trajectory.  The Newton
tolerance is $10^{-10}$
and the maximum Newton count is 30; the symmetric projection uses the adaptive
tolerance $\max(10^{-10},\tau^2)$ consistent with the second-order trajectory
regime of Section~\ref{subsec:projection}, and the projection and midpoint
solves use finite-difference Jacobians.  A few diagnostics deviate
deliberately from these settings: the contact-form residual diagnostics
(Section~\ref{subsec:contact-residual}) tighten the tolerances, the
accumulated variant also replacing the $\tau^2$ projection floor by a $\tau^3$
floor so that the intrinsic $\mathcal{O}(\tau^3)$ residual is exposed rather
than masked, and the active-projection benchmark of
Section~\ref{subsec:exp-active} disables the floor entirely to activate the
Newton solve. The exact per-diagnostic settings are recorded in the
Supplementary Material, \S\,S12, and in the metadata accompanying each result
file.  Timing data are wall-clock medians over three repeats on the same
workstation and should be read as relative within-run timings; as a
machine-independent cost measure we additionally report gradient and
vector-field evaluation counts, modeled as the method's fixed per-step
operation count times the number of steps (counting conventions and the
instrumented double-pendulum comparison in \S\,S12 and Supplementary
Table~S13); they are independent of hardware and of the
finite-difference Jacobian cost, which dominates wall time for the implicit and
projected solves. The generated result files and per-run metadata
backing every table and figure are provided in the accompanying code archive and
regenerated by a single documented driver script; the full benchmark derivations
are given in the Supplementary Material.
Throughout, ``endpoint'' refers to the endpoint
trajectory error against the refined RK4 reference and ``drift'' to the indicated
Hamiltonian or diagonal-constraint drift diagnostic.

\subsection{Direct contact-form residual diagnostic}
\label{subsec:contact-residual}

The diagnostic follows from the coordinate pullback: for a step $\Phi^\tau$ the
one-step conformal defect against the exact law
$(\Phi_H^\tau)^*\eta=e^{-\gamma\tau}\eta$~\eqref{eq:lie_derivative_contact} is
\[
  \rhoeta(\tau)=\bigl\|(D\Phi^\tau(u))^{T}\eta(\Phi^\tau(u))-e^{-\gamma\tau}\eta(u)\bigr\|_2 ,
\]
the discrete counterpart of~\eqref{eq:contact_pres}, with $D\Phi^\tau$ from central
differences ($\varepsilon=10^{-6}$; the long-time oscillator run of
Section~\ref{subsec:exp-longtime} uses $\varepsilon=5\times10^{-3}$, exact there as
the flow is quadratic). Over $\tau\in\{0.08,\dots,0.005\}$ the Pihajoki-contact
residual scales as $\mathcal{O}(\tau^3)$ on all three systems (per-system log-log
slopes $3.00,2.99,3.00$; Supplementary Material, Table~S12(a)). At $\gamma=0$
this is a consistency-rate
diagnostic, not a contact-specific one: any order-$r$ method has
$\rhoeta=\mathcal{O}(\tau^{r+1})$, so implicit midpoint (at the $\sim10^{-11}$ FD
floor) and RK4 (fourth order) attain a smaller one-step residual than the
projected method. The contact-structural distinction is shown separately by the
long-time separable experiment of Section~\ref{subsec:exp-longtime} (zero per-step
defect) and the bounded energy of the symplectic backbone.

To make the distinction from trajectory accuracy more explicit, we also
evaluate the same defect for the full composed map $\Phi^{t}$, with the
Jacobian $D\Phi^{t}$ evaluated at the fixed initial point, for the damped
double pendulum at $T=5$, $\gamma=0.1$, and $\tau=0.05$ (Supplementary
Material, Table~S12(b)). On this finite-time chaotic benchmark the accumulated
residual mixes two effects: the genuine per-step structural defect of the
local consistency estimate~\eqref{eq:contact_pres}, and the amplification of
ordinary trajectory error through $D\Phi^t$; the values must therefore not be
read as a structural ranking. In particular the Pihajoki-contact entry is the
largest in Table~S12(b). Part of this is accounted for by trajectory-error
amplification: RK4 is fourth order and tracks the exact flow, and hence its
conformal Jacobian, more closely at this $\tau$. We must record that this
explanation does not cover the whole table. Implicit midpoint is second-order
accurate, as the projected method is, and its endpoint error on this benchmark
is comparable (slightly larger, in the refinement study of
Section~\ref{subsec:exp-master}), yet its accumulated residual is
$3.00\times10^{-3}$ against the projected method's $8.04\times10^{-2}$,
smaller by a factor of about $27$. An order-counting argument therefore cannot
explain the midpoint column, and we do not have a structural explanation for
it; we report it rather than leave it to be recovered from the table. What we
take from the accumulated residual is consequently only this negative and
narrow reading, and no claim that $\rhoeta$ favors the projected method (on
this diagnostic it does not, against either reference): even for a high-order
trajectory method the contact-form defect is not constrained to vanish, and
over the run it accumulates measurably (up to $2.92\times10^{-4}$ for RK4 and
$8.04\times10^{-2}$ for the lower-order projected trajectory), so $\rhoeta(t)$
detects a structural quantity distinct from endpoint trajectory error. The
cleaner diagnostics are the one-step residual, which removes the
trajectory-error amplification (though its rate is set by each method's
accuracy order, so it reflects consistency rather than a contact-specific
property), and the separable long-time experiment of
Section~\ref{subsec:exp-longtime}, where the splitting is an exact contact map
so that $\rhoeta(t)$ measures structure alone.

The separable damped harmonic oscillator serves only as a sanity check for the
exact separable splitting, the action update, and the contact Hamiltonian
bookkeeping; it is not evidence for the non-separable method, and its validation
is reported in the Supplementary Material. The main test of the proposed method
begins with non-separable Hamiltonians, where the kinetic sub-flow is not an
explicit shear.

\subsection{Core convergence and work precision: the double pendulum}
\label{subsec:exp-doublependulum}

The double pendulum is the principal non-separable test because the mass
matrix~\eqref{eq:massmatrix} couples the two angles through
$\cos(\theta_1-\theta_2)$.  In the conservative case ($\gamma=0$), over the
refinement grid
$\tau\in\{0.005,0.0025,0.00125,0.000625\}$ at $T=1$, the Pihajoki-contact
endpoint trajectory error falls from $2.74\times10^{-4}$ to
$4.28\times10^{-6}$, with consecutive observed orders $2.00$, $2.00$, $2.00$
(Table~\ref{tab:geometry-convergence}).  The
fitted log--log slope over the grid is $2.00$, and the contact Hamiltonian error
decreases at the same second-order rate.  The grid is confined to
$\tau\le5\times10^{-3}$ to keep every point inside the asymptotic regime of this
chaotic benchmark, where the observed order is meaningful; this is a choice of
refinement window and not a stability limit. The stability threshold of the
explicit sub-flow is much coarser and is reported separately in
Section~\ref{subsec:exp-stability}.

This conservative refinement isolates the spatial discretization error, with the
damping switched off so that no action-update error enters; it is the cleanest
test of the second-order rate, but on its own it does not establish convergence
of the damped scheme. We therefore repeat the same refinement on the damped
double pendulum ($\gamma=0.1$, $T=1$), using the identical reproducible driver
with only the friction coefficient changed. The Pihajoki-contact endpoint error
falls from $2.48\times10^{-4}$ at $\tau=0.005$ to $3.87\times10^{-6}$ at
$\tau=6.25\times10^{-4}$, with consecutive observed orders $2.00$, $2.00$, $2.00$,
so the method keeps its second-order convergence with damping active and the
Herglotz action update included. The paper thus reports both a conservative and a
damped convergence check for the double pendulum.

The work-precision benchmark compares endpoint trajectory error against cost for
the conservative double pendulum, using both wall-clock time and the
machine-independent gradient-evaluation count.  RK4 has favorable short-time
order but is not structure preserving.  Implicit midpoint and Pihajoki-contact
are both second order in this test. In the inactive fine-step regime the
projected step carries no Newton correction and is therefore explicit, so it does
less per-step work than any implicit reference: at $\tau=6.25\times10^{-4}$ it
reaches endpoint error $4.28\times10^{-6}$ with $14400/0$ gradient/RHS
evaluations, against $6.03\times10^{-6}$ with $28800/14400$ for implicit midpoint
(work-precision sweep; Supplementary Material). This is the generic explicit-versus-implicit
gap rather than a method-specific advantage: the explicit Tao baseline ties the
projected step on measured gradient evaluations in this regime. The one cost
comparison with implicit midpoint that is structural rather than an
artifact of explicitness, is the analytic-Jacobian active solve reported below,
where the projection solves the smaller $2n$ system in a single Newton iteration
against midpoint's $2n+1$ in two; we make no broader work-precision claim over
implicit midpoint.

Representative error-vs-effort points from this work-precision sweep are
tabulated in the Supplementary Material; effort is measured primarily by
gradient/RHS evaluations, with median wall-clock time as a secondary
implementation-dependent timing.

The dissipative bookkeeping distinguishes mechanical energy from the full
contact Hamiltonian.  In the damped run at $\gamma=0.1$ over $T=6$,
$\Emech(t)$ decreases and the action contribution $\gamma z(t)$ grows as energy
is transferred to the action variable, while the full contact Hamiltonian
$\Hc(t)=\Emech(t)+\gamma z(t)$ tracks $\Hc(0)e^{-\gamma t}$ throughout
(Fig.~\ref{fig:dp-energy-bookkeeping}), with
maximum absolute contact-decay residual $7.50\times10^{-5}$ over the full
interval. The diagnostic is gauge-fixed in the sense of
Section~\ref{sec:reeb-lie}: it is evaluated with the same potential convention
used to generate the dynamics, and the explicit per-system potential
conventions are recorded in the Supplementary Material, \S\,S11.

\begin{figure}[tbp]
\centering
\includegraphics[width=\linewidth]{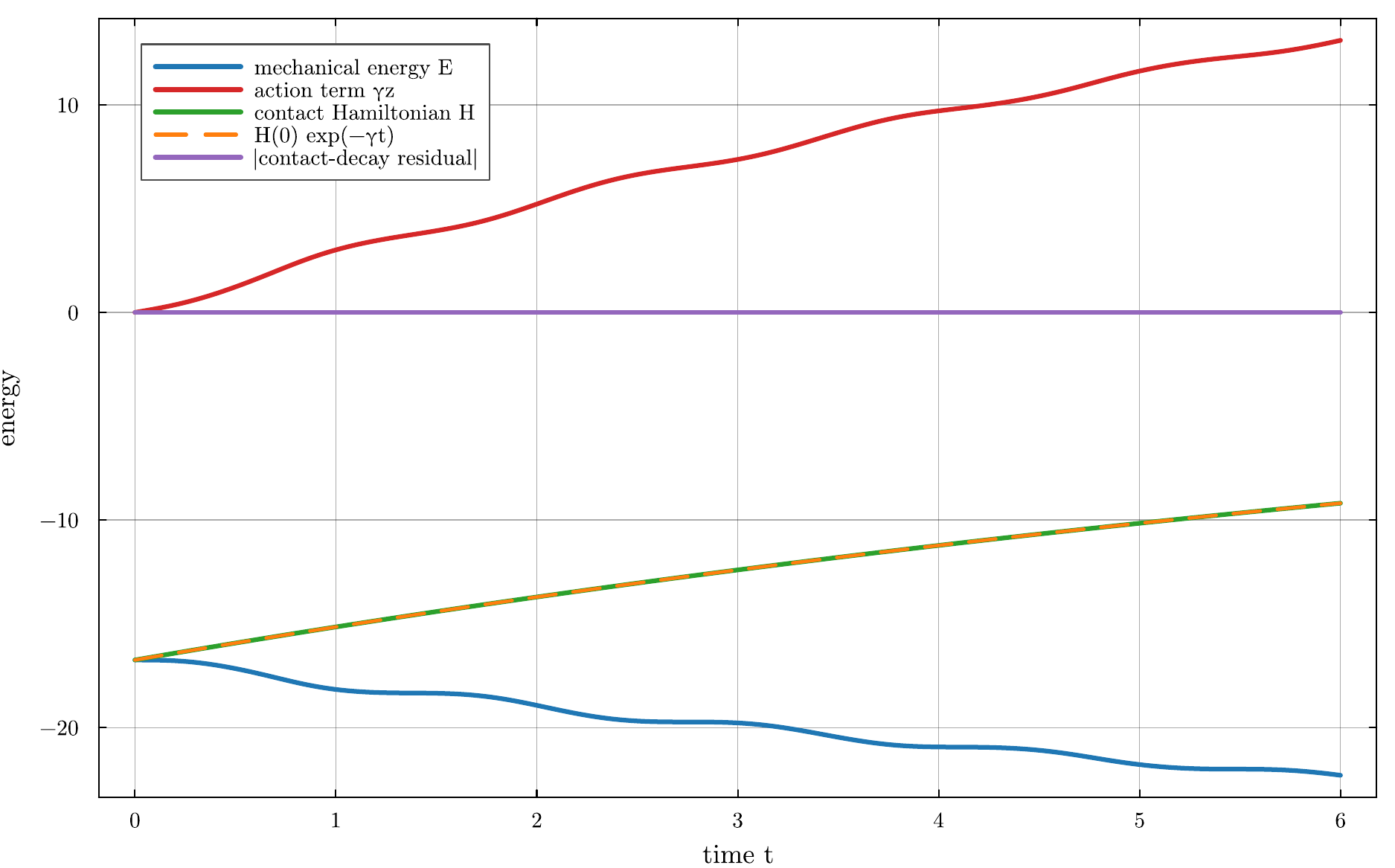}
\caption{Double-pendulum energy bookkeeping at $\gamma=0.1$,
$\tau=0.005$, and $T=6$. The full contact Hamiltonian $\Hc$ follows the expected
decay while $\Emech$ and $\gamma z$ exchange energy.}
\label{fig:dp-energy-bookkeeping}
\end{figure}

\subsection{Core geometric benchmarks: spherical pendulum and torus particle}
\label{subsec:exp-generality}

The spherical pendulum tests a different source of non-separability: the metric
factor $1/\sin^2\theta$ in~\eqref{eq:sphH}, which couples $\theta$ and
$p_\varphi$ and steepens sharply as the trajectory approaches the polar
singularities.  This steepening degrades local accuracy faster than for the
double pendulum, so we start the refinement grid at $\tau=0.005$ to keep every
point in the asymptotic regime; as with the double pendulum, this is a
refinement-window choice and not a stability limit, the measured threshold for
this system being $\tau\approx0.2$ (Section~\ref{subsec:exp-stability}). Over
$\tau\in\{0.005,0.0025,0.00125,0.000625\}$ at $T=1$, the Pihajoki-contact
endpoint error falls from $3.92\times10^{-4}$ to $6.12\times10^{-6}$ with
consecutive observed orders $2.00$, $2.00$, $2.00$, in close agreement with the
slope-two reference (Table~\ref{tab:geometry-convergence}).

With contact damping, $p_\varphi$ is no longer conserved but decays as a
Noether-Herglotz quantity.  This identity is exact by construction of the
discretization rather than a quantitative convergence result: because the
azimuthal angle is cyclic, $\partial_\varphi H=0$, so the kinetic and projection
sub-flows leave $p_\varphi$ untouched and the two damping half-steps rescale it by
$e^{-\gamma\tau}$ exactly. The discrete law $p_\varphi(t)=p_\varphi(0)e^{-\gamma
t}$ therefore holds to roundoff for any scheme of this form, independent of
$\tau$. Consistent with this, over $T=4$ at $\tau=0.0025$ the maximum absolute
deviation is $1.87\times10^{-13}$ for $\gamma=0.1$ and $3.91\times10^{-14}$ for
$\gamma=0.5$ (relative errors $\sim\!10^{-13}$, at the roundoff level), confirming
that no quantity beyond floating-point error enters. We report this as an
exactness property of the splitting for cyclic coordinates, not as a measure of
trajectory accuracy.

To test the Riemannian formulation~\eqref{eq:riemH}, we use a particle on a torus
with major radius $R=3$, minor radius $r=1$, and gravitational potential. The
inverse metric depends on the poloidal angle, and the Gaussian curvature
\[
  K(\theta)=\frac{\cos\theta}{r(R+r\cos\theta)}
\]
changes sign between the outer and inner sides of the torus. The same
Pihajoki-contact formulation applies without any torus-specific modification
beyond supplying $g^{ij}(q)$, $V(q)$, and their derivatives.  Over
$\tau\in\{0.005,0.0025,0.00125,0.000625\}$ at $T=1$, the conservative torus run
has endpoint error decreasing from $1.27\times10^{-4}$ to $1.99\times10^{-6}$
with consecutive observed orders $2.00$, $2.00$, $2.00$
(Table~\ref{tab:geometry-convergence}).  The damped energy-bookkeeping residual for
the torus at $\gamma=0.1$ is $1.08\times10^{-4}$ over $T=6$.

For all three benchmarks Table~\ref{tab:geometry-convergence} separates the full
endpoint error $e_{qpz}$ into its physical component $e_{qp}$ and action component
$e_z$. The physical and action components are each second order (observed order
$2.00$ throughout), so the second-order rate is not an artifact of either part
alone. Their relative size varies with the system: for the double pendulum
$e_{qp}\approx e_z$, for the spherical pendulum $e_z$ is the larger, and for the
torus particle the action component dominates ($e_z$ exceeds $e_{qp}$ by about an
order of magnitude), so there the full-state convergence is driven by the Herglotz
action update rather than by the mechanical variables.

\begin{table}[tbp]
\caption{Pihajoki-contact convergence checks for the non-separable
examples in the conservative case ($\gamma=0$) at $T=1$.  Endpoint errors are
$L^2$ norms of the computed minus reference final state: $e_{qpz}$ over the full
state $(q,p,z)$, $e_{qp}$ over the physical variables $(q,p)$, and
$e_z=|z-z^{\mathrm{ref}}|$ over
the action. ``order'' is the observed order of $e_{qpz}$ between consecutive step
sizes ($e_{qp}$ and $e_z$ are likewise second order; see text). ``$H$ err'' is the
scalar energy-conservation diagnostic $\max_t|\Hc(t)-\Hc(0)|$, a structural
quantity and not a trajectory error. The damped double pendulum ($\gamma=0.1$) is
checked separately in the text.}
\label{tab:geometry-convergence}
\centering
\small
\resizebox{\columnwidth}{!}{
\begin{tabular}{lcccccc}
\toprule
System & $\tau$ & $e_{qpz}$ & $e_{qp}$ & $e_z$ & order & $H$ err \\
\midrule
Double pendulum & 0.005 & $2.74\times10^{-4}$ & $1.93\times10^{-4}$ & $1.95\times10^{-4}$ & -- & $7.28\times10^{-5}$ \\
Double pendulum & 0.0025 & $6.85\times10^{-5}$ & $4.82\times10^{-5}$ & $4.87\times10^{-5}$ & 2.00 & $1.82\times10^{-5}$ \\
Double pendulum & 0.00125 & $1.71\times10^{-5}$ & $1.21\times10^{-5}$ & $1.22\times10^{-5}$ & 2.00 & $4.55\times10^{-6}$ \\
Double pendulum & $6.25\times10^{-4}$ & $4.28\times10^{-6}$ & $3.01\times10^{-6}$ & $3.04\times10^{-6}$ & 2.00 & $1.14\times10^{-6}$ \\
Spherical pendulum & 0.005 & $3.92\times10^{-4}$ & $2.11\times10^{-4}$ & $3.30\times10^{-4}$ & -- & $3.07\times10^{-5}$ \\
Spherical pendulum & 0.0025 & $9.79\times10^{-5}$ & $5.28\times10^{-5}$ & $8.24\times10^{-5}$ & 2.00 & $7.68\times10^{-6}$ \\
Spherical pendulum & 0.00125 & $2.45\times10^{-5}$ & $1.32\times10^{-5}$ & $2.06\times10^{-5}$ & 2.00 & $1.92\times10^{-6}$ \\
Spherical pendulum & $6.25\times10^{-4}$ & $6.12\times10^{-6}$ & $3.30\times10^{-6}$ & $5.15\times10^{-6}$ & 2.00 & $4.80\times10^{-7}$ \\
Torus particle & 0.005 & $1.27\times10^{-4}$ & $1.55\times10^{-5}$ & $1.26\times10^{-4}$ & -- & $1.12\times10^{-4}$ \\
Torus particle & 0.0025 & $3.18\times10^{-5}$ & $3.88\times10^{-6}$ & $3.15\times10^{-5}$ & 2.00 & $2.80\times10^{-5}$ \\
Torus particle & 0.00125 & $7.94\times10^{-6}$ & $9.69\times10^{-7}$ & $7.88\times10^{-6}$ & 2.00 & $7.00\times10^{-6}$ \\
Torus particle & $6.25\times10^{-4}$ & $1.99\times10^{-6}$ & $2.42\times10^{-7}$ & $1.97\times10^{-6}$ & 2.00 & $1.75\times10^{-6}$ \\
\bottomrule

\end{tabular}}
\end{table}

\subsection{Contact and energy-decay diagnostics}
\label{subsec:exp-variable-friction}

The benchmarks above use a constant friction coefficient. The contact-Herglotz
formulation, however, accommodates a position-dependent friction
$\gamma(q)$~\eqref{eq:riemH}, for which the contact Hamiltonian obeys the
path-dependent decay law $\Hc(t)=\Hc(0)\exp(-\!\int_0^t\gamma(q(s))\,ds)$
(eq.~\eqref{eq:riem-energy}). The momentum equation then carries the extra
deflection term $-z\,\partial_q\gamma$, which couples the friction gradient to the
action variable and vanishes identically for constant $\gamma$. We integrate this
case with the exactly integrable variable-friction sub-flow (derived in the
Supplementary Material): holding $q$ fixed over a damping half-step of duration
$h$, the contact generator $\gamma(q)z$ gives in closed form
\begin{equation}
  z\mapsto z\,e^{-\gamma(q)h},\qquad
  p\mapsto e^{-\gamma(q)h}\bigl(p-h\,z\,\partial_q\gamma(q)\bigr),
  \label{eq:vargamma_step}
\end{equation}
composed symmetrically about the conservative projection step.

The $z$-bookkeeping of the variable-friction step differs from the
constant-friction one of Section~\ref{subsec:dissipation}, and we state both
explicitly because they are not interchangeable. For constant $\gamma$ the
damping half-steps leave $z$ untouched and the whole homogeneous factor
$e^{-\gamma\tau}$ is carried by the exponential update~\eqref{eq:z_update}. For
$\gamma(q)$ the roles are reversed: the two half-steps~\eqref{eq:vargamma_step}
each damp $z$ by $e^{-\gamma(q)\tau/2}$, supplying the full factor between them,
and the conservative sub-step contributes the plain midpoint quadrature
\begin{equation}
  z\mapsto z+\tau L_{\mathrm{mid}},
  \label{eq:z_update_vargamma}
\end{equation}
with no homogeneous factor of its own. Applying~\eqref{eq:z_update} inside the
variable-friction composition would damp $z$ twice, by $e^{-2\gamma\tau}$; the
two conventions must not be mixed. The half-step~\eqref{eq:vargamma_step} is
the exact flow of the isolated contact generator $\gamma(q)z$ and is
distinct from the potential-dissipative sub-flow
\eqref{eq:Bflow-variable-z}--\eqref{eq:Bflow-variable-p} of the separable
comparison method: the latter folds the frozen potential $V(q)$ into the same
closed-form solve, whereas here $V$ is carried entirely by the conservative
backbone and only $\gamma(q)z$ is integrated in the damping step. The two are
separate routines in the code (the variable-friction damping half-step versus the
separable $V+\gamma z$ sub-flow) and should not be identified.

We test this on the torus particle with the smooth friction profile
$\gamma(\theta)=\gamma_0(1+a\cos\theta)$, $\gamma_0=0.1$, $a=0.5$, whose gradient
is nonzero in the poloidal angle, so the deflection term is active throughout the
run. Over $T=20$ at $\tau=0.005$ the computed contact Hamiltonian $\Hc(t)$ tracks
the decay law to a maximum absolute residual of $1.32\times10^{-4}$, a relative
residual of $2.6\times10^{-5}$. Against a
refined reference that integrates the full variable-friction contact field,
including the $-z\,\partial_q\gamma$ term, the endpoint error falls from
$4.31\times10^{-4}$ at $\tau=0.01$ to $6.74\times10^{-6}$ at $\tau=0.00125$ with
consecutive observed orders $2.00$. The variable-friction scheme thus keeps the
second-order accuracy of the constant-friction case while reproducing the
path-dependent decay law.

A full per-system contact validation at $\gamma=0.1$ (trajectory overlays,
action, decay law, residuals; endpoint deviations
$5.7\times10^{-4}$--$1.6\times10^{-5}$ across the three systems) is collected
in the Supplementary Material, \S\,S11.

\subsection{Projection activity and diagonal constraint behavior}
\label{subsec:exp-master}
\label{subsec:exp-active}

Under the adaptive tolerance $\max(10^{-10},\tau^{2})$ the symmetric extended step
already meets the constraint on the refined benchmarks, so $\mu=0$ and the step is
effectively explicit (Section~\ref{subsec:projection}). This is specific to the
step size, not the tolerance: the intrinsic $\mathcal{O}(\tau^3)$ gap rises above
the $\tau^2$ floor once the step is coarse enough and the projection then activates
unchanged: on the fast-rotating torus of Section~\ref{subsec:exp-load} the
activation fraction under the production tolerance rises from $0.53$ at
$\tau=0.01$ to $1.00$ for $\tau\ge0.03$, with no failed solves
(Table~\ref{tab:proj-load}). To exhibit the active solve across the full grid
we repeat the conservative runs with the floor disabled and
$\mathrm{tol}=10^{-10}$ (Supplementary Material, Table~S8): on all three benchmark
systems and across the whole step grid the solve converges in exactly two Newton
iterations to a nonzero $\|\mu\|=\mathcal{O}(\tau^3)$, with no failures, driving
the projection residual to the finite-difference floor in one correction; the
Jacobian $\partial G/\partial\mu=4I+\mathcal{O}(\tau)$ is well conditioned by
construction. This is the regime in which the projection is exercised,
and the measured orders there are $2.0000$, $2.0000$, and $1.9999$ for the double
pendulum, spherical pendulum and torus particle respectively: activation changes
neither the endpoint error nor the order, and the exactly-projected hypothesis of
Proposition~\ref{prop:conformal-symplectic} is met rather than bypassed. Its only cost is
the $2n$-evaluation finite-difference Jacobian per iteration (about $49$ gradient
evaluations at $\tau=0.005$ for the double pendulum, versus $9$ inactive and
$29$--$30$ for midpoint, the two midpoint figures being the $\tau=0.005$
work-precision model and the $\tau=0.01$ instrumented count of Table~S13 in
the Supplementary Material), with the structural comparison reported below.

A compact benchmark table in the Supplementary Material (Table~S4)
summarizes representative rows from the quantitative benchmark, read as both an
accuracy table and an applicability map; the reading guide (including why the
separable splitting has no row and why the RK4 error entries are
self-referential) accompanies the table there. The one result from that
comparison used in the sequel is this: across all sampled steps, systems, and
$\gamma\in\{0,0.1,0.5\}$ the projection terminates at its first residual check
($\mu=0$, no failures). The uncorrected gap $\|G(0)\|$ (e.g.\ $1.4\times10^{-5}$,
$2.4\times10^{-5}$, $4.6\times10^{-6}$ for the three systems at $\tau=0.005$,
$\gamma=0$) is already $\mathcal{O}(\tau^3)$, below the $\tau^2$ floor, so the
floor is never binding here.

\subsection{Step-size stability of the explicit sub-flow}
\label{subsec:exp-stability}

The duplicated backbone is explicit, so it carries a step-size restriction in
practice. In the linear model problem the restriction is entirely repaired by
exact projection: for the constant-friction harmonic oscillator of frequency
$\omega$ the duplicated Strang step decouples into two St\"ormer--Verlet maps,
each stable for
$\omega\tau<2$; the unprojected symmetric average has determinant
$1+(\omega\tau)^6/64$; and the projected step conserves the
oscillator energy and is stable at every step size, contracting under
damping by the conformal factor $e^{-\gamma\tau}$ of
Proposition~\ref{prop:conformal-symplectic} (Supplementary Material,
\S\,S15). On the nonlinear benchmarks a finite limit reappears, and the
refinement grids above are reported well inside it. We quantify that limit
rather than assert it. At $\gamma=0$ the bare
energy $H_{\mathrm{bare}}$ is an exact invariant of the continuous system, so a
step can be called stable when the trajectory stays finite over a long horizon,
the projection solve never fails, and the relative bare-energy error stays
bounded; we scan a ladder of steps to $T=50$ and record the coarsest step meeting
this criterion (driver \texttt{analysis/run\_stability\_threshold.jl}).

The measured thresholds are $\tau\approx0.2$ for the double pendulum, the
spherical pendulum and the fast-rotating torus, and $\tau\approx0.4$ for the
torus particle at nominal momentum, with the first unstable step of the ladder at
$0.3$, $0.3$, $0.3$ and $0.5$ respectively. Two consequences follow.
First, the thresholds vary by only a factor of two across the benchmark
geometries, so the stiffness of the metric moves the limit but does not dominate
it in this suite. Second, every step used anywhere in this paper---including the
coarsest stress-test step $\tau=0.08$ of
Section~\ref{subsec:exp-load}---lies below the threshold by at least a factor of
two, and the fine convergence grids by a factor of forty or more. The exclusion
of coarser steps from the convergence grids of
Section~\ref{subsec:exp-doublependulum} is therefore a matter of staying inside
the asymptotic regime, not of instability: on the double pendulum at $T=1$ the
projected step remains finite, solves without failure, and retains its
second-order rate out to $\tau=0.2$. The linear analysis of \S\,S15 shows that
no linear mechanism sets these thresholds, although in the units
$\tau^{*}\omega_{\max}$, with $\omega_{\max}$ the stiffest frozen-coefficient
frequency along the trajectory, the measured brackets straddle the backbone
ceiling $\omega\tau=2$ for three of the four systems. The values reported here
are empirical and specific to these systems and initial data.

\subsection{Stress test: projection controls long-time drift}
\label{subsec:exp-load}

The fast-rotating torus is used here as a mechanism testbed rather than as a
use case: its exact invariant and controllable stiffness let us force the
projection active and isolate its structural effect. On this geometry a
frozen-coordinate exact splitting exists and is the method of choice at
matched cost (Section~\ref{subsec:exp-kevrekidis}); the in-niche comparison
for the projected method is the dense-metric case of
Fig.~\ref{fig:kevrekidis}(b).

The fine-step benchmarks do not by themselves show a structural advantage over the
$\mu=0$ average, since the inter-copy gap already lies below tolerance. The stress
test is a particle in rapid azimuthal circulation on the torus of
Section~\ref{subsec:exp-generality} with $p_\varphi=40$ (bare energy
$H_0\approx57.9$): $p_\varphi$ is conserved and $H_{\mathrm{bare}}$ is an exact
invariant, so its discrete drift is an unambiguous structure diagnostic, and at the
coarse steps used here the gap rises above the $\tau^2$ floor so the projection is
active at every step under the production tolerance (no floor removal).
Integrating the conservative system to $T=10^{3}$ (Table~\ref{tab:proj-load},
Fig.~\ref{fig:proj-load}), at $\tau=0.05$ the unprojected average drifts to relative
energy error $4.3\times10^{-2}$ (growing secularly), while the projected step holds
$3.7\times10^{-4}$ with no secular trend -- two orders of magnitude better, below RK4's
$2.6\times10^{-2}$, and comparable to implicit midpoint ($6.9\times10^{-4}$). The
gap widens as the step coarsens: at $\tau=0.08$ the average has lost the solution
(relative error $66$) while the projected step is still bounded
($1.0\times10^{-3}$), at a mean of two Newton iterations per step. As the step
refines the activation fraction falls ($1.00$ at $\tau\ge0.03$ to $0.53$ at
$\tau=0.01$) and the configurations converge. The flat projected envelope is the
signature of the symplectic conservative backbone (Theorem~5
of~\cite{Jayawardana2023}): the symmetric projection makes the duplicated step
symplectic on the diagonal, so its energy error is bounded whereas the $\mu=0$
average drifts. This advantage is specific to the coarse-step stiff regime; at
still coarser steps ($\tau\gtrsim0.3$) the explicit sub-flow is itself unstable
(Section~\ref{subsec:exp-stability}).

\begin{table}[tbp]
\caption{Projection-controlled regime on the fast-rotating torus
($p_\varphi=40$, $\gamma=0$, $T=10^{3}$): long-time relative energy drift,
projection iterations, and activation fraction.}
\label{tab:proj-load}
\centering
\small
\resizebox{\columnwidth}{!}{
\begin{tabular}{lccccc}
\toprule
$\tau$ & no proj.\ ($\mu=0$) & projected & RK4 & Newton/step & activation \\
\midrule
0.08 & $6.62\times10^{1}$ & $1.05\times10^{-3}$ & $1.31\times10^{-1}$ & 2.00 & 1.00 \\
0.05 & $4.35\times10^{-2}$ & $3.74\times10^{-4}$ & $2.59\times10^{-2}$ & 2.00 & 1.00 \\
0.03 & $2.91\times10^{-3}$ & $1.30\times10^{-4}$ & $2.25\times10^{-3}$ & 2.00 & 1.00 \\
0.02 & $3.80\times10^{-4}$ & $5.69\times10^{-5}$ & $2.99\times10^{-4}$ & 1.91 & 0.91 \\
0.01 & $1.41\times10^{-5}$ & $1.41\times10^{-5}$ & $9.36\times10^{-6}$ & 1.53 & 0.53 \\
\bottomrule

\end{tabular}}
\end{table}

\begin{figure}[tbp]
\centering
\includegraphics[width=\linewidth]{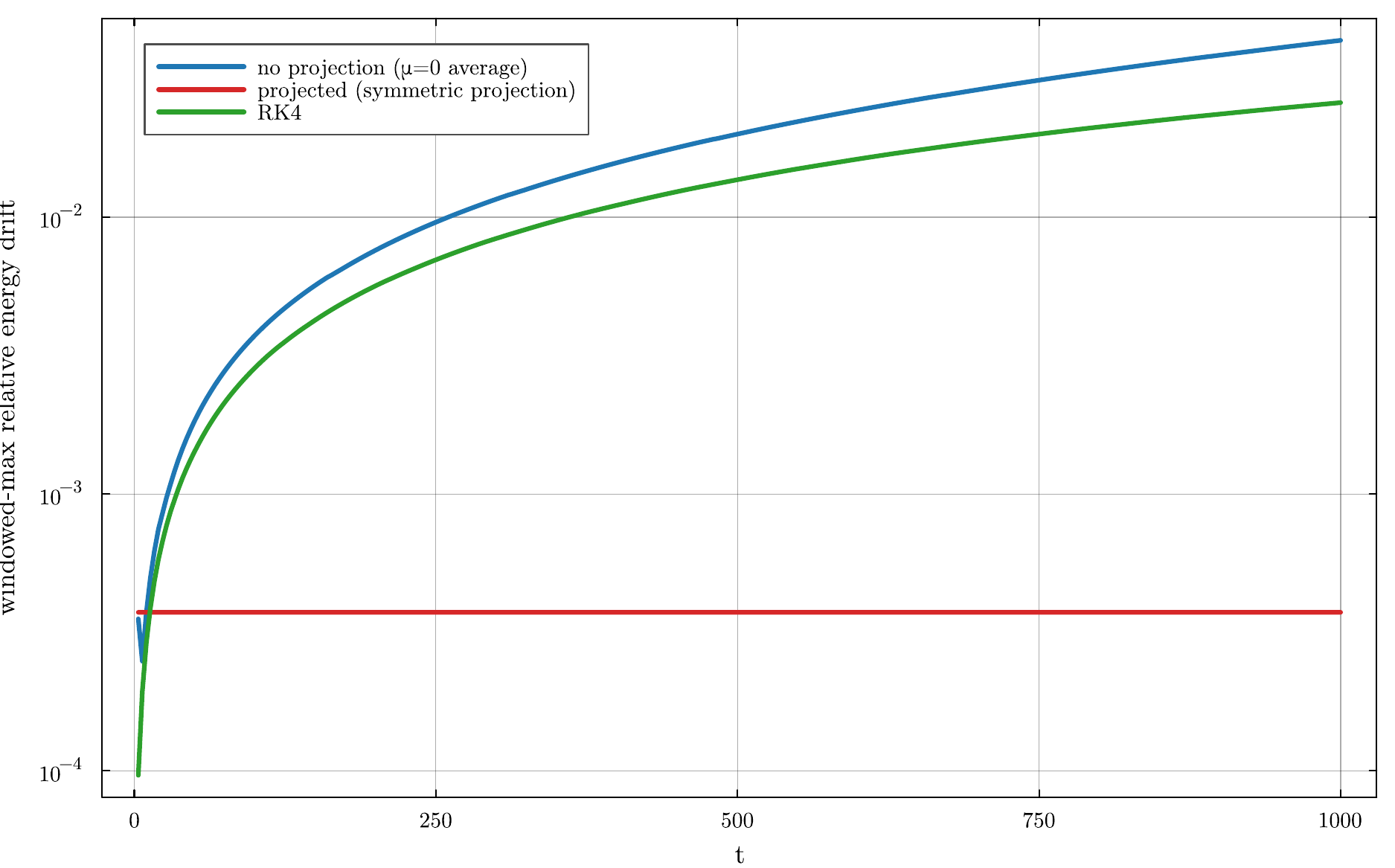}
\caption{Long-time relative energy drift on the fast-rotating torus
($p_\varphi=40$, $\gamma=0$, $\tau=0.05$), shown as the windowed-maximum envelope
of $|H_{\mathrm{bare}}(t)-H_0|/|H_0|$ for the unprojected $\mu=0$ average, the
projected step, and RK4.}
\label{fig:proj-load}
\end{figure}

\subsection{Additional stress tests and ablations}
\label{subsec:exp-activation}

The fine-step benchmarks above show that the method can be effectively explicit: the
symmetric extended step already meets the constraint, so the Newton correction is often
inactive or very small. The experiments here move to active-projection regimes. Two
are conservative ($\gamma=0$): a cadence ablation weakens the projection in a stiff
torus regime, and an ellipsoid test checks activation in a second non-separable
geometry; these probe the projected conservative backbone and the enforcement of the
physical diagonal. A third is dissipative ($\gamma=0.1$), exercising active projection
and damping together so that the contact-decay law is tested alongside the diagonal
return.

A projection-cadence ablation on the same fast torus ($\tau=0.05$, $T=10^3$)
projects only every $k$ steps. Cadences $k=1,2,5$ hold the drift of the per-step
method, while $k=10$ is worse than never projecting at all: by ten uncorrected
steps the multiplier norms reach $0.74$, far outside the local regime
$\mu^*=\mathcal{O}(\tau^3)$ of the analysis, so each late projection is a
large, unmodeled kick. Intermittent projection is not a cheap substitute for
per-step projection. The full ablation, including the distinction between the
single-copy read-off and the $\mu=0$ average, is in the Supplementary Material
(\S\,S14).

\paragraph{Active projection on a triaxial ellipsoid.}
A second activation check on a triaxial ellipsoid ($a=1.0$, $b=1.3$, $c=0.7$),
whose anisotropic induced metric has a nonzero coupling $g_{\theta\varphi}$ and
no cyclic symmetry, confirms the picture in an independent geometry: the
projection activates at every step and enforces the diagonal to tolerance, the
single-copy read-off diverges, and passive averaging matches the projected
accuracy while leaving a one-step gap---the effect is structural, not an
accuracy gain. Details and per-variant data are in the Supplementary Material
(\S\,S14).

\paragraph{Damped active-projection stress test.}
The two checks above are conservative and isolate the projected backbone; to
exercise the projection and the damping together we refine the stiff
fast-rotating torus of Section~\ref{subsec:exp-load} (the conformal-scaling
benchmark, $p_\varphi=40$; as there, a mechanism testbed rather than a use
case) with friction $\gamma=0.1$ and the adaptive $\tau^2$
floor disabled, so the multiplier solve is active at every step (activation
fraction $1.00$, mean two to three Newton iterations, no failed solves). Over
$\tau\in\{0.04,0.02,0.01,0.005\}$ at $T=1$ the full endpoint error $e_{qpz}$
and both of its components $e_{qp}$ and $e_z$ are second order (observed
orders $2.01$, $2.00$, $2.00$), with the action component dominating
($e_z\approx4$--$5\,e_{qp}$), so the $(q,p,z)$ rate is set by the Herglotz
action update; the contact Hamiltonian tracks the decay law
$\Hc(t)=\Hc(0)e^{-\gamma t}$ with a residual that is itself second order in
$\tau$ (full table: Supplementary Material, Table~S11). This is
the dissipative counterpart of the conservative stress test of
Section~\ref{subsec:exp-load}: the second-order trajectory and contact-decay
convergence persist with damping switched on and the projection load-bearing.
Like the conservative tests, it probes the discrete dynamics; it includes no
Tao baseline and is not used to compare against Tao.

\paragraph{Projection-activity summary.}
The projection activity across regimes is collected in a summary table in the
Supplementary Material (Table~S10). In the
fine-step benchmarks the symmetric average already satisfies the constraint to
$\mathcal{O}(\tau^{3})$ below the adaptive tolerance, so the projection is inactive
($\mu=0$, no Newton correction); in the geometrically demanding or coarse-step regimes,
the fast-rotating torus, its cadence ablation, the triaxial ellipsoid, and the
damped fast-rotating torus above, it is
active at every step and structurally relevant. Projection is thus often inactive and
effectively explicit, but becomes a functional, structure-controlling part of the
method exactly where the duplicated dynamics would otherwise drift off the diagonal.

\subsection{Long-time structure diagnostics}
\label{subsec:exp-longtime}

The benchmarks above are reported at the short horizon $T=1$ at which a
fourth-order non-geometric method such as RK4 is most favorable: there its
endpoint trajectory error is the smallest entry in the compact benchmark table
(Supplementary Material, Table~S4), and its one-step contact-form residual is small
as well. Short-time accuracy, however, does not control whether the discrete flow
respects the contact-conformal law over many steps, and it is the long-horizon
behavior that separates a structure-preserving integrator from an accurate but
non-geometric one~\cite{Hairer2006}. This is the long-time counterpart of the
one-step residual of Section~\ref{subsec:contact-residual}: the residual measures
the local conformal defect, whereas the experiment below measures its
accumulation over a long run.

On the one benchmark where structure and accuracy cleanly separate -- the
separable conservative oscillator, for which the contact splitting is a composition
of exact contact maps -- the accumulated contact-form residual
$\rhoeta(t)=\|(D\Phi^{t})^{\top}\eta(\Phi^{t}u_0)-\eta(u_0)\|_2$ stays at the
$\sim\!10^{-13}$ finite-difference floor to $T=2\times10^{3}$ ($\tau=0.05$),
whereas RK4, lacking a discrete conformal law, drifts secularly (from
$\sim\!3\times10^{-8}$ to $\sim\!4.5\times10^{-6}$). This isolates the structural
point -- the splitting carries no contact-form defect at any $\tau$, while a
non-geometric method's defect grows with $t$ -- but it tests the separable
splitting, not the projected method, whose evidence is the one-step estimate and
the run below; on the non-separable benchmarks $\rhoeta(t)$ of the composed map is
instead dominated by trajectory-error amplification, where RK4's higher accuracy
makes its residual smallest (Supplementary Material, Table~S12(b)).

\paragraph{Long-time behavior of the projected method.}
To probe the projected Pihajoki-contact method itself over a long horizon, the
setting in which the duplicated copies could in principle drift apart, we
integrate the non-separable torus particle ($\gamma=0$) to $T=10^{3}$ at
$\tau=0.01$ ($10^{5}$ steps) and track three quantities along the run. The
inter-copy gap $\|(\hat q-\hat x,\,\hat p-\hat y)\|$ measures the separation of the
two copies after the explicit extended step. Because the symmetric projection
re-lifts to the diagonal $q=x$ at the start of every step, this gap is reset each
step and stays at the empirically observed $\mathcal{O}(\tau^3)$ per-step level
for this run
(maximum $3.7\times10^{-5}$, mean $1.6\times10^{-5}$): the copies do not drift
apart. The conservative energy drift $|H_{\mathrm{bare}}(t)-H_{\mathrm{bare}}(0)|$
stays bounded over the full horizon (maximum $4.5\times10^{-4}$, relative
$9.0\times10^{-5}$), consistent with the intended geometric constraint rather
than a secular drift. And the one-step
contact-form residual, sampled along the trajectory, also remains at the
empirical $\mathcal{O}(\tau^3)$ level (maximum $2.2\times10^{-5}$) with no growth in $t$.
This is a direct long-time test of the projected
method, complementing the separable-splitting diagnostic above.

\subsection{Baselines, scaling, and robustness}
\label{subsec:exp-referee}

This subsection carries the explicit Tao baseline into the long-time and
structural tests, identifies a regime in which the projected method, which carries
no Tao binding parameter, compares favorably against the standard rotating Tao
baseline, puts the cost
comparison on a measured footing (including
an analytic-quality-Jacobian data point), and checks second-order behavior and
projection activity at higher dimension, under randomized initial data, and
under finite-difference and tolerance variation, closing with two damped
structural diagnostics. All data are regenerated by the reproducibility driver.

\paragraph{Long-time energy drift with the Tao baseline.}
Section~\ref{subsec:exp-longtime} reported the projected method's own long-time
behavior; here we add the Tao baseline and the unprojected average to the same
stress test, the conservative fast-rotating torus ($p_\varphi=40$, $\gamma=0$),
integrated to $T=10^{3}$, where the bare energy is an exact invariant so the
comparison needs no external reference. Table~\ref{tab:ref-longtime} reports the
maximum relative drift $\max_t|H(t)-H_0|/|H_0|$. The projected method's drift is
bounded and smaller than Tao's (at $\omega=10$) at every step, by roughly an
order of magnitude, and at the coarsest step $\tau=0.08$ Tao diverges
($2.1\times10^{3}$) where the projected method remains controlled
($1.0\times10^{-3}$).

The table does not support a stronger reading than that, and we set out its
limits before drawing on it. Against implicit midpoint the margin is a factor of
$1.6$--$2.0$ at every step---never an order of magnitude---and the advantage is
over Tao and over the unprojected average, that is, over the two constructions
the method is assembled from. Against RK4 the ordering reverses under refinement:
RK4's drift falls below the projected method's at the finest step
($9.4\times10^{-6}$ against $1.4\times10^{-5}$ at $\tau=0.01$), where the
projection has in any case deactivated and the projected and unprojected columns
coincide to three significant figures. The structural claim we make from this
table is therefore confined to the coarse-step regime, where the projection is
active at every step and the drift is bounded with no secular trend; it is not a
uniform accuracy claim, and at fine steps a fourth-order non-geometric method is
the better choice on this diagnostic. The mechanism is the one motivating the construction: the
inter-copy gap $\|(\hat q-\hat x,\,\hat p-\hat y)\|$ opened by the explicit
duplicated step, and what each method does with it afterwards.

The comparison must be stated carefully, because the two methods differ in kind
rather than in magnitude. For the projected step the gap that the explicit
extended map opens before projection is, on the solved constraint,
$\|(\hat q-\hat x,\hat p-\hat y)\|=2\|\mu^*\|$; on this benchmark we measure it
as $1.0\times10^{-1}$, $2.5\times10^{-2}$, $5.3\times10^{-3}$ at
$\tau=0.08,0.05,0.03$, scaling with observed orders $3.05$ and $3.02$, i.e.\ at
the $\mathcal{O}(\tau^3)$ rate of Section~\ref{subsec:exp-master}. This is
not smaller than Tao's gap by any dramatic factor---Tao's is
$6.3\times10^{2}$, $1.1\times10^{-1}$, $3.8\times10^{-2}$ at the same steps, so
away from the divergent coarsest step the two are within a factor of $4$--$7$ of
one another. The difference is that the symmetric average returns the projected
copies to the diagonal exactly, so the gap entering the next step is zero
by construction, whereas Tao's binding rotation leaves its gap in the state and
lets it accumulate; at $\tau=0.08$ it reaches $\mathcal{O}(10^{2})$ and drives
the divergence. It is the exact per-step return, not a smaller instantaneous
gap, that the construction buys.

We flag one artifact of the archive to forestall a misreading: the
\texttt{gap\_projected} column of the accompanying result file records the
projection solver's terminal residual $\|G(\mu^*)\|$, not the inter-copy
gap $2\|\mu^*\|$ quoted above; the two coincide only when $\mu^*=0$. At the two
finest steps of that file the residual also sits on the $\tau^{2}$ tolerance
floor rather than on any intrinsic scale.

\begin{table}[tbp]
\caption{Long-time Tao comparison on the conservative fast-rotating torus
($p_\varphi=40$, $\gamma=0$): maximum relative energy drift
$\max_t|H(t)-H_0|/|H_0|$ at $T=10^{3}$ (``unproj.''\ is the unprojected
symmetric average,
$\mu=0$; Tao at $\omega=10$).}
\label{tab:ref-longtime}
\centering
\small
\resizebox{\columnwidth}{!}{
\begin{tabular}{cccccc}
\toprule
$\tau$ & projected & unproj.\ ($\mu=0$) & RK4 & midpoint & Tao ($\omega=10$) \\
\midrule
$0.08$ & $1.05\times10^{-3}$ & $6.6\times10^{1}$   & $1.31\times10^{-1}$ & $1.71\times10^{-3}$ & $2.14\times10^{3}$ \\
$0.05$ & $3.74\times10^{-4}$ & $4.35\times10^{-2}$ & $2.59\times10^{-2}$ & $6.90\times10^{-4}$ & $3.81\times10^{-3}$ \\
$0.03$ & $1.30\times10^{-4}$ & $2.91\times10^{-3}$ & $2.25\times10^{-3}$ & $2.52\times10^{-4}$ & $1.38\times10^{-3}$ \\
$0.02$ & $5.69\times10^{-5}$ & $3.80\times10^{-4}$ & $2.99\times10^{-4}$ & $1.12\times10^{-4}$ & $5.48\times10^{-4}$ \\
$0.01$ & $1.41\times10^{-5}$ & $1.41\times10^{-5}$ & $9.36\times10^{-6}$ & $2.82\times10^{-5}$ & $1.51\times10^{-4}$ \\
\bottomrule

\end{tabular}}
\end{table}

\paragraph{A regime sensitive to Tao's binding parameter.}
The projected method does not use Tao's binding rotation parameter $\omega$; it
uses instead a projection tolerance, and an active projection may require Newton
iterations, so the trade is one kind of numerical control for another rather
than a parameter eliminated for free. On the same integrable fast torus
($T=200$) we sweep $\omega$ over three decades and compare against Tao's best
drift over all $\omega$, an oracle a user does not have a priori
(Supplementary Material, Table~S14): even against this best case the projected
method is $8$--$10$ times more accurate in long-time energy at every step in
this regime, and the safe $\omega$ window is non-monotonic and step-dependent
($3$ of the $10$ sampled $\omega$ diverge at the two coarser steps), so there
is no uniformly safe choice of $\omega$ in the sampled set. This comparison
supports the projected method in the near-integrable fast-rotation regime,
where the binding rotation is most stressed; it is not a general claim that
projection is more reliable than Tao's method.

\paragraph{Measured cost on an equal footing.}
To make the per-step cost unambiguous we measure gradient (vector) evaluations
by instrumenting every system evaluation, rather than modeling them
(Supplementary Material, Table~S13; $\tau=0.01$, $\gamma=0$). Two facts hold
simultaneously, both consistent with the discussion above: in the inactive
fine-step regime of all headline benchmarks the projected step is markedly
cheaper in gradient work than implicit midpoint ($9$ versus $30$ on the
torus); in the active regime it is more expensive ($49$ versus $30$), since
the finite-difference Jacobian solve dominates exactly where the projection is
load-bearing (Section~\ref{subsec:exp-activation}). Both halves are intrinsic
to the method and are stated together. For the structural comparison we also
build the step Jacobians at analytic quality by complex-step differentiation
on the coupled-rotor chain at a representative active state (Table~S13(b)):
the projection solves a strictly smaller system ($2n$ versus $2n+1$) and
converges in a single Newton iteration against midpoint's two. The two
comparisons pull in opposite directions, and the reconciliation is the
Jacobian construction: the $49$-versus-$30$ deficit is the price of assembling
the projection Jacobian by finite differences ($2n$ extra gradient evaluations
per iteration), while the analytic-Jacobian data show that once the Jacobian
is supplied the solve itself is the smaller one. The finite-difference deficit
is an implementation cost, removable by the analytic tangent maps already
flagged in Section~\ref{sec:discussion}; the structural comparison is the one
that survives it.

\paragraph{Higher-dimensional convergence.}
A coupled-rotor chain with tridiagonal configuration-dependent mass matrix
checks that second-order accuracy and the inactive-projection behavior persist
at $n=4$ and $n=6$ ($2n=12$ projection multipliers); all methods retain their
orders and the projection stays inactive in the fine-step regime, exactly as in
the low-dimensional benchmarks (Supplementary Material, \S\,S13). This is a
robustness check at higher dimension, not a scalability claim.

\paragraph{Damped non-separable diagnostics.}
Two damped diagnostics complete the comparison; the full data are in the
Supplementary Material, Table~S6. On the damped torus ($\gamma=0.2$) RK4 has
the smallest accumulated contact-form residual $\rhoeta(t)$, so the projected
method does not win on $\rhoeta$ on the non-separable damped target and we make
no contact-residual claim against RK4; this is consistent with the reading of
Section~\ref{subsec:contact-residual}, where $\rhoeta$ is a consistency
diagnostic dominated by trajectory-error amplification, not a structural
ranking. On the fast torus, where the cyclic coordinate makes
$p_\varphi(t)=p_\varphi(0)e^{-\gamma t}$ exact, the projected method reproduces
this conformal momentum decay to roundoff ($2.0\times10^{-13}$) whereas Tao
corrupts it to $5.1\times10^{-5}$ through its binding rotation, which mixes
$p_\varphi$ with the copy variable and never returns them to the diagonal. This
isolates the structural defect of the Tao read-off, but it is not evidence for
the contact methods over an accurate non-geometric integrator: RK4 reaches
$4.6\times10^{-14}$, marginally below the projected method, so RK4 wins the
momentum decay taken by itself. The contact-structural advantage of the
projected method appears in the long-time conformal energy behavior
(Table~\ref{tab:ref-longtime}), not in a smaller one-step or accumulated
$\rhoeta$ and not in the momentum decay.

\paragraph{Randomized robustness.}
\label{subsec:exp-ensemble}
To check that the second-order accuracy and inactive-projection behavior
survive away from the nominal state, we draw $50$ initial conditions per
benchmark in a $\pm 25$ percent bare-energy shell (fixed seed) and estimate the
observed order over $\tau\in\{0.005,0.0025,0.00125\}$. Every accepted state of
every system gives mean observed order $2.00$ (zero spread to two decimals)
with median energy drift at the $10^{-7}$--$10^{-5}$ level; the
projection-activation fraction at the coarsest sampled step is state-dependent
(per-run maxima up to $0.29$ on the spherical pendulum), so its per-system mean
should not be read as a bound (full statistics in the Supplementary Material,
\S\,S11).

\paragraph{Finite-difference and tolerance ablation.}
The active projection is insensitive to its two numerical parameters: sweeping
the finite-difference step over six orders of magnitude leaves the converged
multiplier unchanged to ten significant figures, and tightening the tolerance
below $10^{-8}$ saturates the residual at the roundoff floor, so the tolerance
bounds the converged residual rather than being tracked by it. The full sweeps
are in the Supplementary Material (\S\,S14). The active-regime results above
are therefore not artifacts of either setting.

\subsection{Head-to-head with exact-contactomorphism splittings}
\label{subsec:exp-kevrekidis}

This section compares the projected method directly against splittings that preserve the contact structure exactly. We implement two constructions in
the strict/prolonged generator class of~\cite{Kevrekidis2026}; the driver is
\texttt{analysis/run\_kevrekidis\_comparison.jl} and the module is
\texttt{Simulations/KevrekidisSplitting.jl} in the reproducibility package. The depth-one realization below is our own implementation, assembled from the published text of~\cite{Kevrekidis2026}, whose demonstrations are one-degree-of-freedom; its measured performance reflects what can currently be built from that text, not the best the framework might achieve with the higher-order gadget coefficients that its Proposition~4.3 leaves unspecified.

\paragraph{Frozen-coordinate metrics: the splitting exists and wins at matched
cost.} For the torus particle and the spherical pendulum the three-generator
splitting of Section~\ref{sec:nonsep} is available in closed form. We verify
that it converges at second order on both systems, conservative and damped
(observed orders $2.000$--$2.001$ across
$\tau\in[2.5\times10^{-3},2\times10^{-2}]$), and that its one-step pullback
satisfies $\|\Phi^*\eta-e^{-\gamma\tau}\eta\|\lesssim10^{-10}$ (the
finite-difference floor) at every tested $\tau$ and $\gamma$: the
conformal-contact identity holds exactly, where the projected method attains it
only to $\mathcal{O}(\tau^3)$. On the fast-torus stress test of
Section~\ref{subsec:exp-load} the comparison is two-sided. At equal step the
projected method's drift constant is an order of magnitude smaller (a factor of
$12$ at $\tau=0.05$: $4.6\times10^{-3}$ for the splitting against
$3.7\times10^{-4}$ projected; a factor of $11$ at $\tau=0.08$). At matched cost the ordering reverses: the splitting
needs no Newton solve, no duplication and no Jacobian, so running it at
$\tau=2.5\times10^{-3}$ still costs less than the projected step at
$\tau=0.05$ and achieves drift $1.1\times10^{-5}$, thirty-three times below the
projected method's $3.7\times10^{-4}$, with $\eta$ preserved exactly
throughout. Where a frozen-coordinate splitting exists, it, and not the
projected method, is the right choice. Both panels of
Fig.~\ref{fig:kevrekidis} summarize the comparison.

\begin{figure}[tbp]
\centering
\includegraphics[width=\linewidth]{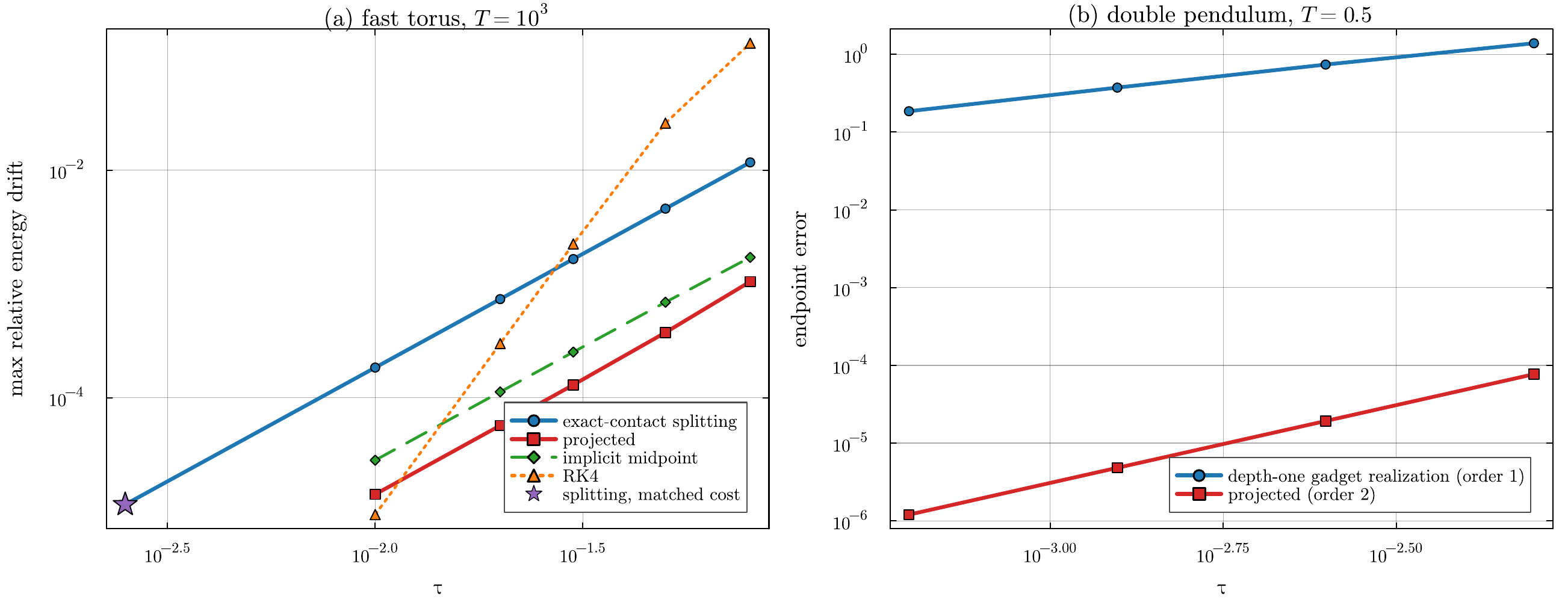}
\caption{Head-to-head with exact-contactomorphism splittings.
\textbf{(a)}~Fast-rotating torus ($p_\varphi=40$, $\gamma=0$, $T=10^3$): maximum
relative energy drift against $\tau$ for the three-generator exact-contact
splitting of Section~\ref{sec:nonsep},
the projected method, implicit midpoint, and RK4. At equal $\tau$ the projected
method's constant is an order of magnitude smaller ($11$--$12\times$ across the
shared steps); the starred point is the
splitting at its matched-cost step $\tau=2.5\times10^{-3}$ (no Newton solve, no
duplication), which undercuts the projected method's entire curve while
preserving $\eta$ exactly. \textbf{(b)}~Double pendulum ($\gamma=0$, $T=0.5$,
common refined-RK4 reference): endpoint error for our depth-one gadget
realization of the commutator construction (measured order $\approx1$, the
expected rate for this realization) against
the projected method (order $2$). The two maps trade opposite virtues: the
depth-one realization preserves the conformal pullback of $\eta$ exactly (to the
finite-difference floor) but is inaccurate, while the projected method is
accurate but carries its $\mathcal{O}(\tau^3)$ contact residual. On this dense
metric the accuracy gap---four to five orders of magnitude---is far larger than
the structural one.}
\label{fig:kevrekidis}
\end{figure}

\paragraph{Dense metrics: the realizable splitting alternative is not
competitive.} For the double pendulum no term-wise splitting with elementary
sub-flows exists in its chart
(Section~\ref{sec:nonsep}), and the splitting route requires the depth-one
commutator representation of~\cite{Kevrekidis2026}. We construct it explicitly:
the kinetic term is written as
$\{F_1,p_1^3/3\}+\{F_m,p_1^2p_2/2\}+\{F_4,p_2^3/3\}$ with closed-form
antiderivatives ($F_1'=a+b/2$, $F_m'=b$, $F_4'=-c$ for the coefficient
functions $a,b,c$ of the double-pendulum metric; the representation is
verified to reproduce $T$ to
machine precision), each bracket flow realized by symmetric group-commutator
gadgets with $\sqrt{h}$ sub-steps and the whole step palindromically
symmetrized. Every sub-step is an exact contactomorphism, and the composed step
indeed satisfies the exact conformal pullback to the finite-difference floor
($\lesssim10^{-9}$) for $\gamma\in\{0,0.1\}$. Its accuracy, however, is not
competitive: the realization converges at first order with measured error
$\approx3\times10^{2}\,\tau$ on the benchmark orbit ($1.9\times10^{-1}$ at
$\tau=6.25\times10^{-4}$, observed orders $0.91$--$1.01$), so matching the
projected method's endpoint error at that step would require a step several
orders of magnitude smaller. First order is the expected rate of this
realization, not an implementation artifact. The $\sqrt{h}$ sub-steps are the
mechanism: each gadget makes internal excursions of size $\mathcal{O}(\sqrt h)$
whose cancellation is exact for the contact form but only approximate for the
trajectory; in the Baker--Campbell--Hausdorff expansion of the palindromic
gadget the half-power excursion terms cancel under reversal of the $\sqrt h$
sub-steps (palindromy in the sub-step parameter), leaving an
$\mathcal{O}(h^{2})$ one-step trajectory defect that accumulates to first
order globally. Higher-order gadget realizations are asserted to exist
(Proposition~4.3 of~\cite{Kevrekidis2026}) but their composition coefficients
are not specified in the published text, and to our knowledge no multi-DOF
realization has appeared.

\paragraph{Dense metrics: the implicit exactly-contact alternative.}
The splitting route is not the only exactly-contact one. The lifted implicit
midpoint of Section~\ref{sec:nonsep}, Strang-composed with the exact
damping-potential flow, is second order and exactly conformal on any dense
metric, at the price of a Newton solve of the full $2n$-dimensional kinetic
map in every step. We implement it and run it on the damped double pendulum
under the protocol of this section (driver \texttt{run\_lifted\_midpoint.jl}).
It behaves as constructed: the measured $\eta$-pullback defect stays at the
finite-difference floor ($\lesssim2\times10^{-10}$ across
$\tau\in\{0.08,0.05,0.01\}$, with no $\tau^{3}$ scaling, for
$\gamma\in\{0,0.1\}$), the observed trajectory order is $2.000$, and the
contact decay law is reproduced ($3\times10^{-3}$ relative decay error at
$\tau=0.08$, $\gamma=0.1$, $T=6$). On the horizon of the panel the projected
method is more accurate at every step size (endpoint error $7.7\times10^{-5}$
against $1.5\times10^{-4}$ at $\tau=5\times10^{-3}$, a factor of about $1.9$
down the refinement ladder). The dense-metric comparison is therefore not
about existence: exactly-contact second-order alternatives can be built, and
the implicit one preserves $\eta$ itself where the projected step preserves
$\omega=d\eta$ exactly and carries an $\mathcal{O}(\tau^3)$ residual in
$\eta$. It is about cost and structure per step: a full-state implicit solve
with exact $\eta$ on one side; a projection-confined solve, inactive in the
fine-step regime, with exact $\omega$ rescaling and the smaller measured
trajectory error on the other.

\FloatBarrier

\section{Discussion and conclusions}
\label{sec:discussion}\label{sec:limitations}

The experiments support the central point of the paper: contact splitting is
effective for separable dissipative Hamiltonians, but it does not solve the
generic non-separable problem created by configuration-dependent kinetic energy.
The double pendulum, spherical pendulum, and torus particle all fall outside the
separable regime, yet their contact-Herglotz structure is handled uniformly by
the Pihajoki-contact method, with second-order accuracy throughout and a
projection that is inactive in the fine-step benchmarks but structurally relevant
in the coarse-step stiff regime of Section~\ref{subsec:exp-load}; the active-regime
probes of Section~\ref{subsec:exp-activation} are not evidence for a global
damped contact-preservation theorem. Relative to the contact splitting work of
Vijayan et al.~\cite{Vijayan2025}, whose setting is a separable or
Lie-structured problem in which exact sub-flows compose directly, the present
focus is complementary: here the kinetic flow is a nonlinear geodesic flow,
exact splitting is unavailable, and the projected extended-phase-space strategy
serves as the conservative backbone, extended with contact damping and a
Herglotz action update.

The sharper contrast is with Kevrekidis~\cite{Kevrekidis2026}. As noted in
Section~\ref{sec:introduction}, the class treated here lies exactly inside the
polynomial-in-$p$ Lie algebra of that work: the metric $g^{ij}(q)$ enters as a
coefficient function rather than as an obstruction, so no approximation step is
incurred, and its sub-steps are exact contactomorphisms. On our own target class
it therefore delivers exact contact preservation where we obtain an
$\mathcal{O}(\tau^3)$ one-step residual, with $C^r$ error control under the
hypothesis that the flow remains in a compact set. We do not claim a theoretical
advantage over it. The practical
question---which trade wins where---is no longer left open: the head-to-head of
Section~\ref{subsec:exp-kevrekidis} resolves it in both directions. Where the
metric has frozen-coordinate structure ($\mathbb{T}^2$, $S^2$), the exact
splitting in that generator class wins at matched cost, by more than an order of
magnitude in long-time drift, while preserving $\eta$ exactly; the projected
method's advantage there is confined to equal-step comparisons. Where the metric
is dense (the double pendulum), the depth-one realization we could assemble from
the published text is first order with a prohibitive constant, while the lifted
implicit midpoint of Section~\ref{sec:nonsep} is second order and exactly
conformal at the cost of a full-state Newton solve in every step; between the
two the projected method holds the semiexplicit middle ground, second order
with the smaller measured trajectory error and a solve confined to the
projection. The niche this paper occupies is that middle ground on dense
metrics, a cost-and-structure trade rather than a uniqueness claim. The most
informative experiment now open on the explicit side is a higher-order
multi-degree-of-freedom realization of the commutator construction, which
would sharpen the trade further.

The method has clear trade-offs, and its scope is deliberately limited, as set
out once in the Remark following Proposition~\ref{prop:contact-consistency}:
the analysis is local, one-step, and constant-friction, and does not establish
global contact preservation of the one-form $\eta$. The method is not fully
explicit: a symmetric projection solve remains part of the step, acting only on
the $2n$ projection variables. In the fine-step benchmarks the multiplier stays
zero and no correction is applied across the sampled steps, systems, and
friction values, but this inactive behavior is specific to the tested regime
and should not be assumed at coarser steps or in higher dimension: in
Section~\ref{subsec:exp-load} the projection activates at every step and
removing it produces secular energy drift. The double-pendulum work-precision
comparison therefore supports a cost advantage over implicit midpoint only in
the inactive-projection regime and for this finite-difference implementation;
against the Tao extended phase-space baseline with binding rotation, the
projected method is competitive in accuracy and also returns the duplicated
copies exactly to the physical diagonal at each step. The explicit duplicated
sub-flow also imposes a step-size stability threshold, measured in
Section~\ref{subsec:exp-stability} at
$\tau\approx0.2$--$0.4$ across the benchmark geometries, and all results are
reported at least a factor of two inside it. That threshold is a nonlinear
phenomenon: the linear model problem is unconditionally stable under exact
projection (Supplementary Material, \S\,S15), and none of our estimates
controls it.
Position-dependent
friction $\gamma(q)$ is implemented through an exactly integrable
variable-friction sub-flow and reproduces the path-dependent decay law
numerically, but its analytical one-step estimate remains open; nonlinear
dependence on $z$ lies outside the present construction. The benchmarks are
low-dimensional (up to three degrees of freedom, with a coupled-rotor robustness
check at $n=4,6$), so the finite-difference projection Jacobians should be
replaced by analytic tangent maps before substantially higher-dimensional tests.
Extending the contact-conformal estimate to $\gamma(q)$, analytic Jacobians,
high-dimensional benchmarks, adaptive time stepping, and comparison
against contact variational and Herglotz-based integrators~\cite{Maciel2023} are
left for future work. The low-dimensional evidence supports the method as a
practical projected route for non-separable contact Hamiltonian systems, not as a
high-dimensional scalability result or a universal replacement for implicit
methods.

\section*{CRediT authorship contribution statement}
\textbf{Lorena Loera-Galeana:} Conceptualization, Formal analysis,
Investigation, Writing -- original draft.
\textbf{Santiago Mej\'ia:} Formal analysis, Investigation, Software,
Validation, Visualization, Writing -- review \& editing.
\textbf{Espartaco Alvarado:} Formal analysis, Investigation, Software,
Validation, Visualization, Writing -- review \& editing.
\textbf{H\'ector Medel-Cobaxin:} Conceptualization, Formal analysis,
Investigation, Methodology, Supervision, Validation, Visualization,
Writing -- original draft.

\section*{Declaration of competing interest}
The authors declare that they have no known competing financial interests or
personal relationships that could have appeared to influence the work reported
in this paper.

\section*{Funding}
This research did not receive any specific grant from funding agencies in the
public, commercial, or not-for-profit sectors.

\section*{Data availability}
The code, benchmark parameters, scripts, metadata, and reproducibility files
supporting this study are available from the corresponding author upon
reasonable request.
The package contains source code, input parameters, plotting scripts,
comma-separated result files, per-run metadata, and the Julia project
environment used for the computations. A single documented driver,
\texttt{analysis/run\_all.jl}, regenerates every result file, generated
\LaTeX{} table body, and figure used in the main text and the Supplementary
Material; the one randomized experiment (the robustness ensemble) uses a fixed
seed (\texttt{20260618}) recorded in its run metadata. Results were produced
with Julia~1.12.6, which we recommend for reproducing them; Julia~1.12 or later
is required, since under earlier versions a closure-capture miscompilation can
cause the finite-difference projection Jacobian to vanish silently and yield
incorrect projections. The requirement is declared in the package
\texttt{Project.toml}. Upon journal acceptance an immutable snapshot of the
package will be archived on a DOI-minting repository and the DOI cited in the
published version. The full
derivations of the benchmark Lagrangians, Hamiltonians, equations of motion,
and dissipation laws are provided in the accompanying Supplementary Material.

\section*{Acknowledgements}
The authors have no acknowledgements to report.

\section*{Declaration of generative AI and AI-assisted technologies in
  the manuscript preparation process}
During the preparation of this work the authors used AI-assisted tools
(Claude and Consensus) to improve the readability and language of the
manuscript, and to assist with literature searches. The authors take
full responsibility for the content of the published article.

\end{document}